\documentclass{article}

\usepackage{microtype}
\usepackage{graphicx}
\usepackage{subfigure}
\usepackage{booktabs} 
\usepackage{multirow}
\usepackage{subcaption}
\usepackage{algorithmicx}
\usepackage{algpseudocode}

\usepackage{hyperref}

\usepackage[accepted]{icml2025}

\usepackage{amsmath}
\usepackage{amssymb}
\usepackage{mathtools}
\usepackage{amsthm}

\usepackage[capitalize,noabbrev]{cleveref}

\theoremstyle{plain}

\theoremstyle{definition}

\theoremstyle{remark}

\usepackage{thmtools, thm-restate}

\usepackage[textsize=tiny]{todonotes}

\icmltitlerunning{PharmacophoreBridge}

\begin{document}

\twocolumn[
\icmltitle{Pharmacophore-guided de novo drug design with diffusion bridge}



\icmlsetsymbol{equal}{*}

\begin{icmlauthorlist}
\icmlauthor{Conghao Wang}{ccds}
\icmlauthor{Jagath C. Rajapakse}{ccds}
\end{icmlauthorlist}

\icmlaffiliation{ccds}{College of Computing and Data Science, Nanyang Technological University, Singapore}

\icmlcorrespondingauthor{Jagath C. Rajapakse}{ASJagath@ntu.edu.sg}

\icmlkeywords{Diffusion bridge \and Geometric learning \and Drug design \and Generative model}

\vskip 0.3in
]



\printAffiliations{}  

\begin{abstract}
\textit{De novo} design of bioactive drug molecules with potential to treat desired biological targets is a profound task in the drug discovery process. Existing approaches tend to leverage the pocket structure of the target protein to condition the molecule generation. However, even the pocket area of the target protein may contain redundant information since not all atoms in the pocket is responsible for the interaction with the ligand. In this work, we propose PharmacoBridge, a phamacophore-guided \textit{de novo} design approach to generate drug candidates inducing desired bioactivity via diffusion bridge. Our method adapts the diffusion bridge to effectively convert pharmacophore arrangements in the spatial space into molecular structures under the manner of $SE(3)$-equivariant transformation, providing sophisticated control over optimal biochemical feature arrangements on the generated molecules. PharmacoBridge is demonstrated to generate hit candidates that exhibit high binding affinity with potential protein targets. 
\end{abstract}

\section{Introduction}
\label{submission}

Computer-aided drug design (CADD) plays a crucial role in the modern drug discovery procedure. Conventional CADD approaches, such as virtual screening, require searching for candidates with optimal molecular properties in a vast chemistry library. Although accelerated by the high-throughput technology, this process is still time-consuming and costly \cite{wouters2020estimated} since the relationship between chemical structures and the molecular property of interest is obscure. 
\textit{De novo} design, on the other hand, models the chemical space of molecular structures and properties and seeks for the optimal candidates in a directed manner \cite{meyers2021novo} instead of enumerating every possibility, thus facilitating the drug discovery process. Moreover, the flourishing of deep generative models in various domains, such as large language models and image synthesis, has endowed us an opportunity of applying deep learning to improving \textit{de novo} drug design algorithms. 

Generative models including variational autoencoder (VAE) \cite{kingma2013auto}, generative adversarial networks (GAN) \cite{goodfellow2020generative} and denoising diffusion probabilistic models (DDPM) \cite{ho2020denoising}, have been successfully adapted for molecular design. Initially, researchers tend to represent drugs with linear notations, such as Simplified Molecular-Input Line-Entry System (SMILES) \cite{weininger1988smiles}, due to its simplicity. Then long-short term memory (LSTM) networks were readily applied to encoding the SMILES notations, and VAE and GAN algorithms were adapted for generation \cite{prykhodko2019novo, mendez2020novo, born2021paccmannrl, kaitoh2021triomphe}. Such methods, however, suffered from low chemistry validity rates of generated molecules since the structural information is neglected in SMILES notations. 

To alleviate such limitations, a wide range of molecular graph-based generative algorithms has emerged. For example, graph VAE \cite{simonovsky2018graphvae} paved the way for probabilistic graph construction approaches, which can be used as the generator of other generative models such as GANs \cite{de2018molgan}. However, due to computational complexity, such methods were only able to generate a limited number of nodes. A couple of research managed to address this issue by proposing auto-regressive graph construction algorithms. Flow-based methods such as MoFlow \cite{zang2020moflow} and GraphAF \cite{shi2020graphaf} modeled the generation of bonds and atoms in a sequential decision flow and generated them in order. MoLeR \cite{maziarz2021learning} and MGSSL \cite{zhang2021motif} introduced motif-based generation that considers functional groups rather than single atoms as nodes, thus expanding the scale of generated molecules. Auto-regressive generation also improved the validity rate of generated molecules since the valence check can be performed at every step of generation. However, such methods are unnatural since there is no such decision sequence in a molecule, where posterior nodes depend on former generated ones. In addition, if one step in the middle of the generation is predicted improperly, the whole subsequent generation will be affected. 

The advent of diffusion models \cite{song2020score} enables molecular graph generation in one-shot. Although diffusion models initially obtained tremendous success in text and image generation \cite{rombach2022high}, a wealth of research has demonstrated that they can be adapted for graph generation as well since they can be developed to learn the distributions of the adjacency matrix and the feature matrix, which fully define the graph \cite{vignac2022digress, jo2022score, lin2024functional}. In the meantime, researchers realized the power of diffusion models is not restricted on 2D graph generation and can be evolved into 3D point cloud generation, which allows for more sophisticated design of molecules. In 2D graph generation, subsequent tests over all 3D conformers have to be conducted to filter the candidate since different isomers of a molecule may exhibit various pharmaceutical properties, which may lead to additional costs and future attrition. Aiming at achieving direct 3D spatial design of molecules, a series of studies were proposed \cite{xu2022geodiff, xu2023geometric, schneuing2022structure, hoogeboom2022equivariant, guan20233d}. Equivariant geometric learning were leveraged to ensure the roto-translational equivariance of their systems while generating coordinates of the nodes. 

Nevertheless, latest research indicates that the typical diffusion models, such as score matching with Langevin dynamics (SMLD) \cite{song2019generative} and DDPM \cite{ho2020denoising}, can be generalized as \textit{diffusion bridges} that are guaranteed to reach an endpoint from the desirable domain in pre-fixed time \cite{liu2022let, ye2022first} according to Doob's \textit{h}-transform \cite{doob1984classical}. Moreover, it is demonstrated that diffusion bridges not only can be applied to unconditional generation by mapping data distribution to the prior noise distribution but also can be leveraged to align two arbitrary distributions by fixing both the start point and the end point \cite{chen2021likelihood, vargas2021solving, wang2021deep}. Diffusion bridges have been successfully employed for image translation \cite{su2022dual, zhou2023denoising}, physics-informed molecule generation \cite{wu2022diffusion, jo2023graph}, molecular docking \cite{somnath2023aligned} and equilibrium state prediction \cite{luo2024bridging}. We believe that they are promising approaches to hit molecule design as well since the start point can be naturally considered as the molecular structures and the end point as desired conditions, e.g., pharmacophore arrangements. 

In the field of drug design, increasing research has delved into adapting generative models to the design of hit candidates with potential to react with particular biological targets. For example, gene expression profiles were used to condition the generation so that the generated molecules would lead to desirable biological activities \cite{mendez2020novo, born2021paccmannrl, kaitoh2021triomphe, wang2024gldm}. However, gene expression profiles may contain redundant information, since not all genes are related to the drug reaction. To enable more sophisticated control over the generation procedure, researchers proceeded to consider the target protein structure as the constraint. Pocket2Mol \cite{peng2022pocket2mol} first proposed an auto-regressive model to generate atoms and bonds gradually under the guidance of the protein pocket structure. Further research improved pocket-based generation by one-shot generation using diffusion models, such as TargetDiff \cite{guan20233d}. D3FG \cite{lin2024functional} extended atom-based generation to functional groups, aiming at preserving better chemical consistency.

Recent studies proposed to generate hit candidates guided by pharmacophores, which hypothesizes a spatial arrangement of chemical features that are essential for the binding between the drug molecule and the target protein. Unlike conditioning by the protein pocket structure, which depends on the model to learn to discover the indispensable features for binding implicitly, pharmacophore arrangements explicitly define such features, which renders them suitable constraints for hit identification. Nonetheless, existing research studying pharmacophore-guided drug discovery, e.g., PharmacoNet \cite{seo2023pharmaconet}, only leveraged pharmacophore modelling to compute matching scores between ligand and protein, and thereby accelerated virtual screening. Although PGMG \cite{zhu2023pharmacophore} introduced a \textit{de novo} pharmacophore-guided drug design framework, it only learned a latent variable of the pharmacophore features and generated SMILES representations accordingly. This actually destructs the 3D features given by the pharmacophore hypotheses and omits their connection with the 3D molecular structures. In this work, we propose a direct pharmacophore-guided \textit{de novo} drug design approach that learns to transform the pharmacophore distribution into the molecular structure distribution with an $SE(3)$-equivariant diffusion bridge. 

\section{Preliminaries}
\label{sec:bg}


\begin{figure*}[ht]
    \centering
    \includegraphics[width=0.9\linewidth]{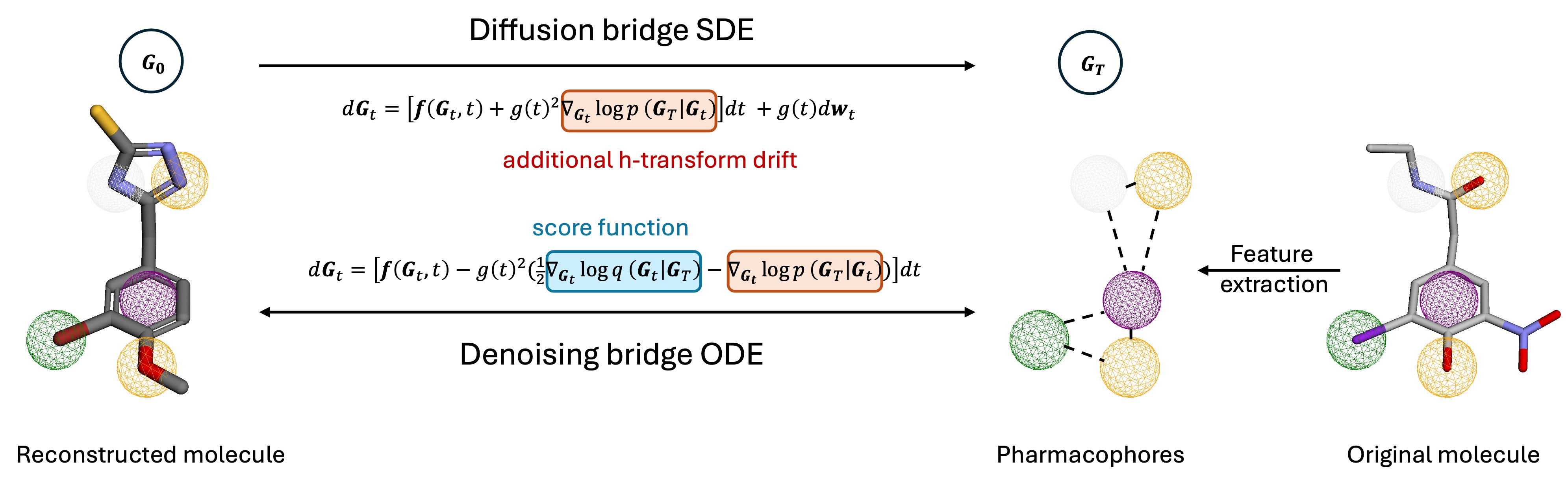}
    \caption{Overview of PharmacoBridge. The diffusion bridge process is devised to map the molecule data $\mathbf{G}_0$ to the extracted pharmacophore data $\mathbf{G}_T$ via Doob's $h$-transforms. Reversely, a score matching model is trained to estimate the score function, which composes the denoising bridge process that recovers molecule data from the pharmacophore data. }
    \label{fig:diagram}
\end{figure*}

\subsection{Notations and Problem Formulation}

Drug molecules and their pharmacophore features are represented as point clouds in the 3D Euclidean space, i.e., $\mathbf{G} = ( \mathbf{x}, \mathbf{h} )$, where $\mathbf{x} \in \mathbb{R}^{N \times 3}$ and $\mathbf{h} \in \mathbb{R}^{N \times M}$ denote atom coordinates and features, respectively. 
As shown in Figure \ref{fig:diagram}, our mission is to learn the evolution of the atom point cloud over time $\{\mathbf{G}_t\}_{t=0}^T$ via a bridge process $q( \mathbf{G}_t | \mathbf{G}_T )$ that allows us to sample molecular structures from pharmacophore hypotheses. 

Let the initial state $\mathbf{G}_0 \sim q_{data}(\mathbf{g})$ denote the molecular point cloud and the terminal state $\mathbf{G}_T \sim q_{data}(\mathbf{\Gamma})$ denote the point cloud of the corresponding pharmacophore features. Since molecules $\mathbf{g}$ and pharmacophores $\mathbf{\Gamma}$ are strictly paired data, rotations or translations to one of them will cause equivariant transformations to the other, i.e., $q_{data}(g \cdot \mathbf{g}, g \cdot \mathbf{\Gamma}) = q_{data}(\mathbf{g}, \mathbf{\Gamma})$ $\forall g \in SE(3)$, which indicates an $SE(3)$-invariant joint distribution. Consequently, the evolution of the system should satisfy $q(g \cdot \mathbf{G}_0 | g \cdot \mathbf{G}_T) = q(\mathbf{G}_0 | \mathbf{G}_T)$ $\forall g \in SE(3)$, indicating an $SE(3)$-equivariant transition density. 

\subsection{Diffusion Model with SDEs}

A diffusion process is to inject noise into data $\mathbf{G}_0 \sim q_{data}$ and gradually convert data into noise $\mathbf{G}_T \sim p_{prior}$, where $q_{data}$ and the $p_{prior}$ denote the data and the prior distribution, respectively. Song et al. proposed to model this diffusion process as the solution to the following SDE \cite{song2020score}
\begin{equation}
    \label{eq:diff_proc}
    d\mathbf{G}_t = \mathbf{f}(\mathbf{G}_t, t)dt + g(t)d\mathbf{w}_t
\end{equation}
where $\mathbf{f}(\cdot, t)$ is the drift function, $\mathbf{w_t}$ is the standard Wiener process and $g(\cdot)$ is the diffusion coefficient. In our case, we denote the representation of a molecule at time $t$ as $\mathbf{G}_t$. 

The denoising process is the solution to the reverse-time SDE
\begin{equation}
    \label{eq:denoi_proc}
    d\mathbf{G}_t = [\mathbf{f}(\mathbf{G}_t, t) - g(t)^2 \nabla_{\mathbf{G}_t}\log p(\mathbf{G}_t)]dt + g(t)d\mathbf{w}_t
\end{equation}
where $p(\mathbf{G}_t)$ is the marginal distribution of $\mathbf{G}_t$ at time $t$. SMLD \cite{song2019generative} and DDPM \cite{ho2020denoising} can be regarded as such SDE frameworks with different designs of $\mathbf{f}(\cdot, t)$ and $g(\cdot)$.

\subsection{Diffusion Bridge}

Doob's $h$-transform \cite{doob1984classical} reveals that by adding an additional drift term, the diffusion process is guaranteed to reach a fixed end point from an arbitrary distribution, given by
\begin{equation}
    \label{eq:diff_bridge}
    d\mathbf{G}_t = [\mathbf{f}(\mathbf{G}_t, t) + g(t)^2 \nabla_{\mathbf{G}_t}\log p(\mathbf{G}_T | \mathbf{G}_t)]dt + g(t)d\mathbf{w}_t
\end{equation}
where $p(\mathbf{G}_T | \mathbf{G}_t)$ is the transition density function of the original diffusion process solving \eqref{eq:diff_proc}. Starting from $\mathbf{G}_0 \sim q_{data}(\mathbf{g})$, the score function $\nabla_{\mathbf{G}_t}\log p(\mathbf{G}_T | \mathbf{G}_t)$ drives the process to reach the desired endpoint $\mathbf{G}_T = \mathbf{\Gamma} \sim q_{data}(\mathbf{\Gamma})$. 


\section{Pharmacophore to Drug Diffusion Bridge}
\label{sec:mtd}

We introduce the pharmacophore to drug diffusion bridge model (PharmacoBridge) to translate 3D arrangements of pharmacophore features into hit molecules. The proposed framework is presented in Figure \ref{fig:diagram}. 

\subsection{Equivariant Denoising Diffusion Bridge}

We present a graph translation paradigm that generates 3D point cloud guided by prior cluster arrangements. We demonstrate our approach on the task of 3D hit molecule design guided by the pharmacophore hypothesis. Based on the diffusion bridge characterized by Eq. \eqref{eq:diff_bridge}, the time-reversed SDE of the denoising bridge is defined in the following theorem.
\begin{restatable}[Denoising diffusion bridge]{theorem}{denoisingbridge}
    \label{thm:denoisingbridge}
    By reversing the $h$-transformed diffusion bridge with the law of \eqref{eq:diff_bridge}, we obtain the following ODE to model the time evolution of the transition density $q(\mathbf{G}_t | \mathbf{G}_T)$:
    \begin{equation}
        \label{eq:denoi_bridge}
        \begin{aligned}
            d\mathbf{G}_t = [ &\mathbf{f}(\mathbf{G}_t, t) - g^2(t) (\frac{1}{2} \nabla_{\mathbf{G}_t}\log q(\mathbf{G}_t | \mathbf{G}_T) \\
            &- \nabla_{\mathbf{G}_t}\log p(\mathbf{G}_T | \mathbf{G}_t))]dt 
        \end{aligned}
    \end{equation}
    which constructs the denoising bridge process $\{\mathbf{G}_t\}_{t=0}^T$ with marginal distribution $\mathbf{G}_T = \mathbf{\Gamma} \sim q_{data}(\mathbf{\Gamma})$, inducing the joint distribution $q(\mathbf{G}_0, \mathbf{G}_T)$ approximating $q_{data}(\mathbf{g}, \mathbf{\Gamma})$. $p(\mathbf{G}_T | \mathbf{G}_t)$ is the transition density associated with the original diffusion process specified by Eq. \eqref{eq:diff_proc}. 
\end{restatable}


The proof of deriving the denoising bridge process specified in \cref{thm:denoisingbridge} is provided in Appendix \ref{sec:denoi_bridge_proof}. Next, we aim to devise an $SE(3)$-equivariant diffusion bridge that satisfies symmetry constraints. As demonstrated in \cite{luo2024bridging}, given a diffusion process solving SDE \eqref{eq:diff_proc}, its associated transition density $p(\mathbf{G}_T | \mathbf{G}_t)$ is $SE(3)$-equivariant, i.e., $p(g \cdot \mathbf{G}_T | g \cdot \mathbf{G}_t) = p(\mathbf{G}_T | \mathbf{G}_t)$, $\forall g \in SE(3)$,
if (1) the initial density $q(\mathbf{G}_0)$ is $SE(3)$-invariant; and (2) $\mathbf{f}(\cdot, t)$ is $SO(3)$-equivariant and $T(3)$-invariant; (3) the transition density of $\mathbf{w}_t$ is $SE(3)$-equivariant. Based on this equivariant diffusion process, we further devise the equivariant denoising diffusion bridge:
\begin{restatable}[Equivariant denoising diffusion bridge]{theorem}{eddb}
    \label{thm:eddb}
    
    Given an $SE(3)$-equivariant diffusion process with transition density $p(\mathbf{G}_T | \mathbf{G}_t)$, its $h$-transformed diffusion bridge with the law of \eqref{eq:diff_bridge} and the corresponding denoising diffusion bridge characterized in \cref{thm:denoisingbridge} also preserve $SE(3)$-equivariance. Their transition densities satisfy
    \begin{itemize}
        \item $p(g \cdot \mathbf{G}_t | g \cdot \mathbf{G}_0, g \cdot \mathbf{G}_T) = p(\mathbf{G}_t | \mathbf{G}_0, \mathbf{G}_T)$, $\forall g \in SE(3)$
        \item $q(g \cdot \mathbf{G}_t | g \cdot \mathbf{G}_T) = q(\mathbf{G}_t | \mathbf{G}_T)$, $\forall g \in SE(3)$
    \end{itemize}
\end{restatable}
    
The constructed equivariant denoising diffusion bridge enables us to map molecule and pharmacophore distributions retaining symmetry constraints. Proofs of \cref{thm:eddb} are provided in Appendix \ref{sec:eddb_proof}. 

\paragraph{Denoising score matching}


The training objective is designed to estimate the score $\nabla_{\mathbf{G}_t}\log q(\mathbf{G}_t | \mathbf{G}_T)$ in order to solve the denoising bridge ODE \eqref{eq:denoi_bridge}
\begin{equation}
    \label{eq:training_obj}
    \begin{aligned}
        \mathcal{L}_{\theta} = &\mathbb{E}_{(\mathbf{G}_0, \mathbf{G}_T) \sim q_{data}(\mathbf{g}, \mathbf{\Gamma}), \mathbf{G}_t \sim q(\mathbf{G}_t | \mathbf{G}_0, \mathbf{G}_T), t \sim \mathcal{U}(0, T)} \\
        &\left[\lambda(t) \left \| s_{\theta}(\mathbf{G}_t, \mathbf{G}_T, t) - \nabla_{\mathbf{G}_t}\log q(\mathbf{G}_t | \mathbf{G}_0, \mathbf{G}_T) \right \|^2 \right] 
    \end{aligned}
\end{equation}
where the initial and final states are fixed at $\mathbf{g}$ and $\mathbf{\Gamma}$, respectively.
$\lambda(t)$ is a time-dependent loss weight to intensify the penalty as the denoising bridge approaches the data point, improving the reconstruction at later stages. Details of the choice of this term are discussed in Appendix \ref{sec:scaling}. 

The true score of the transition density $q(\mathbf{G}_t | \mathbf{G}_T) = \int_{\mathbf{G}_0} q(\mathbf{G}_t | \mathbf{G}_0, \mathbf{G}_T) q_{data}(\mathbf{G}_0 | \mathbf{G}_T) d\mathbf{G}_0$ is intractable in general. We therefore train a neural network model $s_{\theta}(\mathbf{G}_t, \mathbf{G}_T, t)$ to approximate the score of the sampling distribution $\nabla_{\mathbf{G}_t}\log q(\mathbf{G}_t | \mathbf{G}_0, \mathbf{G}_T)$. 
Specific parameterization strategies are discussed in the next section. 
DDBM \cite{zhou2023denoising} has shown that the optimal solution of \eqref{eq:training_obj} approximates the true score $s_{\theta}^*(\mathbf{G}_t, \mathbf{G}_T, t) = \nabla_{\mathbf{G}_t}\log q(\mathbf{G}_t | \mathbf{G}_T)$.

\subsection{Score Model Parameterization}

To parameterize the score model $s_{\theta}(\mathbf{G}_t, \mathbf{G}_T, t)$, we design the sampling distribution $q(\mathbf{G}_t | \mathbf{G}_0, \mathbf{G}_T)$ to be the same as the diffusion distribution fixed on both ends, which is a Gaussian distribution known in closed-form
\begin{equation}
    \label{eq:q_dist}
    \begin{aligned}
        q(\mathbf{G}_t | \mathbf{G}_0, \mathbf{G}_T) &:= p(\mathbf{G}_t | \mathbf{G}_0, \mathbf{G}_T) = \mathcal{N}(\hat{\mu}_t, \hat{\sigma}_t^2 \mathbf{I}), \\
        \mathrm{where} \quad \hat{\mu}_t & = \frac{\text{SNR}_T}{\text{SNR}_t} \frac{\alpha_t}{\alpha_T} \mathbf{G}_T + \alpha_t \mathbf{G}_0 (1 - \frac{\text{SNR}_T}{\text{SNR}_t}), \\
        \mathrm{and} \quad \hat{\sigma}_t^2 &= \sigma_t^2 (1 - \frac{\text{SNR}_T}{\text{SNR}_t})
    \end{aligned}
\end{equation}
where $\alpha_t$ and $\sigma_t$ are the signal and noise schedules, which vary with the choice of the bridge design. The signal-to-noise ratio is thus defined as $\text{SNR}_t = \alpha_t^2 / \sigma_t^2$.

Since denoising bridge distribution is designed to be Gaussian, its score function can be readily derived by 
\begin{equation}
    \label{eq: score_func}
    \nabla_{\mathbf{G}_t}\log q(\mathbf{G}_t | \mathbf{G}_0, \mathbf{G}_T) = - \frac{\mathbf{G}_t - \hat{\mu}_t}{\hat{\sigma}_t^2}
\end{equation}
We train $s_{\theta}(\mathbf{G}_t, \mathbf{G}_T, t)$ to estimate the score above, which can be reparameterized as
\begin{equation}
    \label{eq:score_model}
    \begin{aligned}
    &s_{\theta}(\mathbf{G}_t, \mathbf{G}_T, t) = -\frac{\mathbf{G}_t - \mu_{\theta}}{\hat{\sigma}_t^2}, \\
    \mathrm{where} \quad &\mu_{\theta} = ( \frac{\text{SNR}_T}{\text{SNR}_t} \frac{\alpha_t}{\alpha_T} \mathbf{G}_T + \alpha_t D_{\theta}(\mathbf{G}_t, \mathbf{G}_T, t) (1 - \frac{\text{SNR}_T}{\text{SNR}_t}) ) \\
    \end{aligned}
\end{equation}
where the denoiser $D_{\theta}(\mathbf{G}_t, \mathbf{G}_T, t)$ aims at predicting $\mathbf{G}_0$. Following EDM \cite{karras2022elucidating}, $D_{\theta}$ can be further reparameterized with 
\begin{equation}
    \label{eq:edm}
    D_{\theta}(\mathbf{G}_t, \mathbf{G}_T, t) = c_{skip}(t) \mathbf{G}_t + c_{out}(t) F_{\theta} (c_{in}(t) \mathbf{G}_t, c_{noise}(t))
\end{equation}
where $F_{\theta}(\cdot, t)$ is the neural network that predicts something between $\mathbf{G}_0$ and the noise at time step $t$. In typical diffusion models \cite{song2020score, song2019generative, ho2020denoising}, $F_{\theta}(\cdot, t)$ aims at directly predicting the noise at time step $t$, but this is demanding when $t$ approaches the final step $T$, since the data is almost full of noise. EDM proposed to adapt $F_{\theta}$ to predict the mixture of signal and noise to alleviate this drawback \cite{karras2022elucidating}. Specific parameterizations for VP- and VE-based diffusion bridges are illustrated in Appendix \ref{sec:bridge_para}. The choice of the scaling parameters $c_{in}$, $c_{out}$, $c_{skip}$ and $c_{noise}$ is discussed in Appendix \ref{sec:scaling}. 

To preserve $SE(3)$-equivariance, we implement $F_{\theta}(\cdot, t)$ with EGNN \cite{satorras2021n}. The behavior of $F_{\theta}$ at step $t$ is formulated as
\begin{equation}
    \label{eq:f_theta}
    \begin{aligned}
        (\Tilde{\mathbf{x}}, \Tilde{\mathbf{h}}) &= F_{\theta} (c_{in}(t) \mathbf{G}_t, c_{noise}(t)) \\
        &= F_{\theta} (c_{in}(t) (\mathbf{x}_{t}, \mathbf{h}_{t}), c_{noise}(t))
    \end{aligned}
\end{equation}
where $\mathbf{x}$ and $\mathbf{h}$ denote the atom positions and features, respectively. The output $\mathbf{\Tilde{x}}$ and $\mathbf{\Tilde{h}}$ represent the mixture of noise and signal. The reconstructed initial data is calculated via $\hat{\mathbf{x}}_0 = c_{skip}(t) \mathbf{x}_t + c_{out}(t) \Tilde{\mathbf{x}}$ and $\hat{\mathbf{h}}_0 = c_{skip}(t) \mathbf{h}_t + c_{out}(t) \Tilde{\mathbf{h}}$ according to Eq. \eqref{eq:edm}. Implementation details of $F_{\theta}$ are provided in Appendix \ref{sec:se3_model}. 

\paragraph{Training objective}
Practically, the loss function in Eq. \eqref{eq:training_obj} is equivalent to
\begin{equation}
    \label{eq:new_obj}
    \mathcal{L}_{\theta} = \mathbb{E}_{\mathbf{x}_t, \mathbf{x}_0, \mathbf{h}_t, \mathbf{h}_0, t} (\| \hat{\mathbf{x}}_0 - \mathbf{x}_0 \|^2 + \omega_h \| \hat{\mathbf{h}}_0 - \mathbf{h}_0 \|^2)
\end{equation}
which is the combination of the MSE losses with regard to $\mathbf{x}$ and $\mathbf{h}$. The weight of the feature loss $\omega_h$ is empirically configured as 10 to keep the training stable. 


\subsection{Sampling with PharmacoBridge}

\begin{algorithm*}[ht]
    \caption{Deterministic sampler of PharmacoBridge}\label{alg:sampler}
    \begin{algorithmic}
        \State \textbf{Input: } sampled time steps $t_{i \in {N, ..., 0}}$, pharmacophore data $\mathbf{G}_N \sim p_{pp}$, score matching model $\mathbf{s}_{\theta}$
        \For{$i=N:1$} 
            \State $\mathbf{d}_i \leftarrow - \mathbf{f}(\mathbf{G}_i, t_i) + g^2(t_i) ( \frac{1}{2} \mathbf{s}_{\theta}(\mathbf{G}_i, \mathbf{G}_N, T) - \nabla_{\mathbf{G}_i}\log p(\mathbf{G}_N | \mathbf{G}_i)) )$
            \Comment Solving Eq. \ref{eq:denoi_bridge}
            \State $\mathbf{G}_{i-1} \leftarrow \mathbf{G}_{i} + (t_{i-1} - t_i) \mathbf{d}_i$
            \Comment{Euler step}
            \If {$i \neq 1$}
                \State $\mathbf{d}_i' \leftarrow - \mathbf{f}(\mathbf{G}_{i-1}, t_{i-1}) + g^2(t_{i-1}) ( \frac{1}{2} \mathbf{s}_{\theta}(\mathbf{G}_{i-1}, \mathbf{G}_N, T) - \nabla_{\mathbf{G}_{i-1}}\log p(\mathbf{G}_N | \mathbf{G}_{i-1})) )$
                \State $\mathbf{G}_{i-1} \leftarrow \mathbf{G}_{i} + (t_{i+1} - t_i) (\frac{1}{2} \mathbf{d}_i + \frac{1}{2} \mathbf{d}_i')$
                \Comment{Heun's 2nd order correction}
            \EndIf
        \EndFor
        \State \textbf{return} $\mathbf{G}_0$
    \end{algorithmic}
\end{algorithm*}


Alg. \ref{alg:sampler} illustrates the deterministic sampling procedure of our model. The output of the score matching model $\mathbf{s}_{\theta}$ is computed via Eq. \eqref{eq:score_model} and \eqref{eq:edm}. 
The $h$-transform drift $\nabla_{\mathbf{G}_i}\log p(\mathbf{G}_N | \mathbf{G}_i)$ is derived by calculating the score of the transition probability from $t_i$ to $t_N$. This computation is tractable when using the VE and VP bridge designs, since both of them employs Gaussian transition kernels. Our time discretization strategy is discussed in \cref{sec:time_disc}. 

\section{Experiments}

\subsection{Experimental Setup}

\paragraph{Data}
CrossDocked2020 V1.3 dataset \cite{francoeur2020three} is used for model development and evaluation. The dataset is filtered by preserving the ligands with intimate protein binding poses (RMSE $<$ 0.1 \AA), which results in around 242 thousand ligands. CrossDocked2020 data has organized the ligand and protein files according to pockets. We thus split the dataset in the same way, ensuring molecules with similar structures or biological targets occur either in the training or the sampling dataset. 80\% of the pockets are used for training, 10\% for validation and 10\% for testing, resulting in around 197 thousand molecules in the training set, 19 thousand molecules in the validation set and 15 thousand molecules in the testing set. Details of data preprocessing are illustrated in \ref{sec:data_preproc}.

\paragraph{Baselines}

We benchmark proposed PharmacoBridge against various state-of-the-art 3D molecule design approaches. In the \textit{unconditional generation} task, we compare our method with EDM \cite{hoogeboom2022equivariant} and GruM \cite{jo2023graph}. EDM employed the typical DDPM scheme with an $E(3)$-equivariant GNN backbone for 3D molecule generation. GruM devised its generative process as a destination-predicting diffusion mixture, but never explored conditional generation by fixing the initial point. 

In the \textit{pharmacophore-guided generation} task, we compare our method with Pocket2Mol \cite{peng2022pocket2mol} and TargetDiff \cite{guan20233d}. Pocket2Mol leveraged auto-regressive generation for pocket-guided molecule design. TargetDiff adapted DDPM by incorporating pocket information into the denoising backbone with the attention mechanism. Candidates generated by both methods exhibited potential binding affinity with the target protein pockets. Another recent pharmacophore-guided molecule generation approach, PGMG \cite{zhu2023pharmacophore}, is not considered since PGMG generated SMILES rather than 3D structures, which is fundamentally different from our approach. 

\subsection{Unconditional Generation}

We tested our model on the unconditional generation task to evaluate several chemical and physical properties of the generated molecules. Specifically, 10,000 molecules were sampled using our model or each baseline. Validity, novelty, uniqueness, synthetic accessibity (SA) and quantitative estimate of drug-likeness (QED) were assessed on the generated molecules. 

\begin{table}[ht]
    \centering
    \caption{Unconditional generation performance of PharmacoBridge and baselines. }
    \resizebox{\linewidth}{!}{
    \begin{tabular}{ccc|ccc}
        \toprule
        \multicolumn{3}{c|}{} & Validity (\%) & Uniqueness (\%) & Novelty (\%) \\ \midrule
        \multicolumn{3}{c|}{EDM} & \textbf{100.00} & 6.89 & 100.00 \\
        \multicolumn{3}{c|}{GruM} & \textbf{100.00} & 82.60 & 100.00 \\ \midrule
        \multirow{4}{*}{PharmacoBridge} & \multirow{2}{*}{VE} & Basic & 96.63 & 85.37 & 100.00 \\
         &  & Aromatic & 88.27 & 90.90 & 100.00 \\
         & \multirow{2}{*}{VP} & Basic & 99.91 & 90.72 & 99.98 \\
         &  & Aromatic & 99.96 & \textbf{91.94} & \textbf{100.00} \\ \bottomrule
    \end{tabular}
    }
    \label{tab:uncond}
\end{table}


\begin{figure*}[ht]
    \centering
    \includegraphics[width=0.95\textwidth]{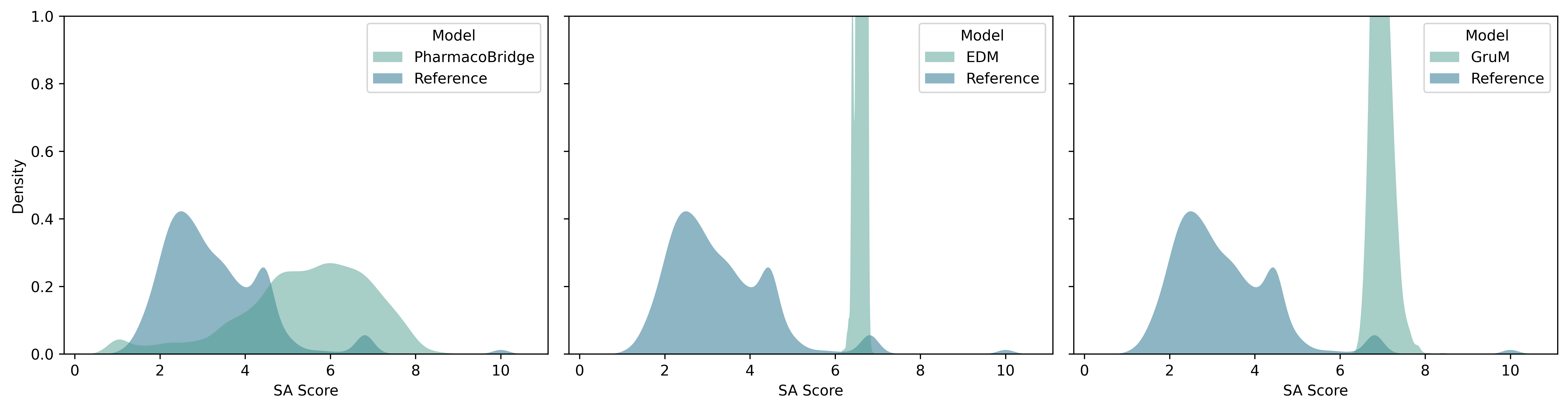}
    \caption{SA score ($\downarrow$) distribution. SA scores of the molecules generated by EDM and GruM are concentrated between 6.0 and 8.0, which is obviously larger than the distribution of the original dataset. The SA scores achieved by our model are lower than the baselines and more evenly distributed like the original dataset distribution. }
    \label{fig:sa}
\end{figure*}

\begin{figure*}[ht]
    \centering
    \includegraphics[width=0.95\textwidth]{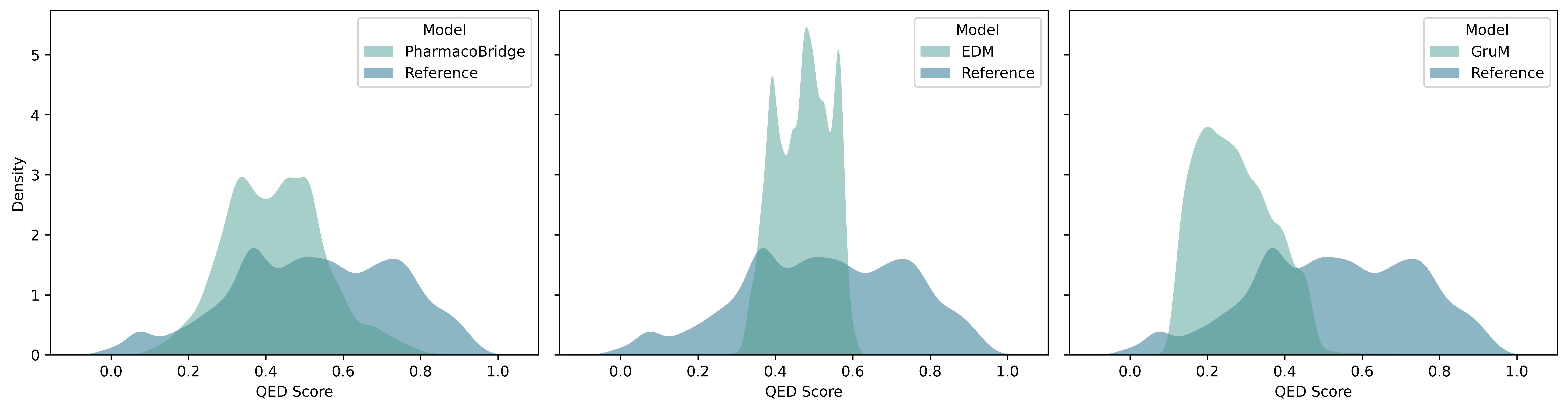}
    \caption{QED score ($\uparrow$) distribution. Molecules generated by our model and EDM exhibit higher QED scores than GruM. The distribution achieved by our model shows similar pattern to the dataset distribution. }
    \label{fig:qed}
\end{figure*}

As shown in Table \ref{tab:uncond}, EDM and GruM exhibited 100\% validity and novelty, but tended to produce repeated samples, resulting in relatively low uniqueness. In the ablation study to explore the effects of different bridge types and feature types, we found that VP bridge-based models outperformed VE-based ones, especially in terms of validity. VP-based PharmacoBridge with aromatic features achieved the highest uniqueness and nearly 100\% validity and novelty. 

\begin{figure*}[ht]
    \centering
    \includegraphics[width=0.95\textwidth]{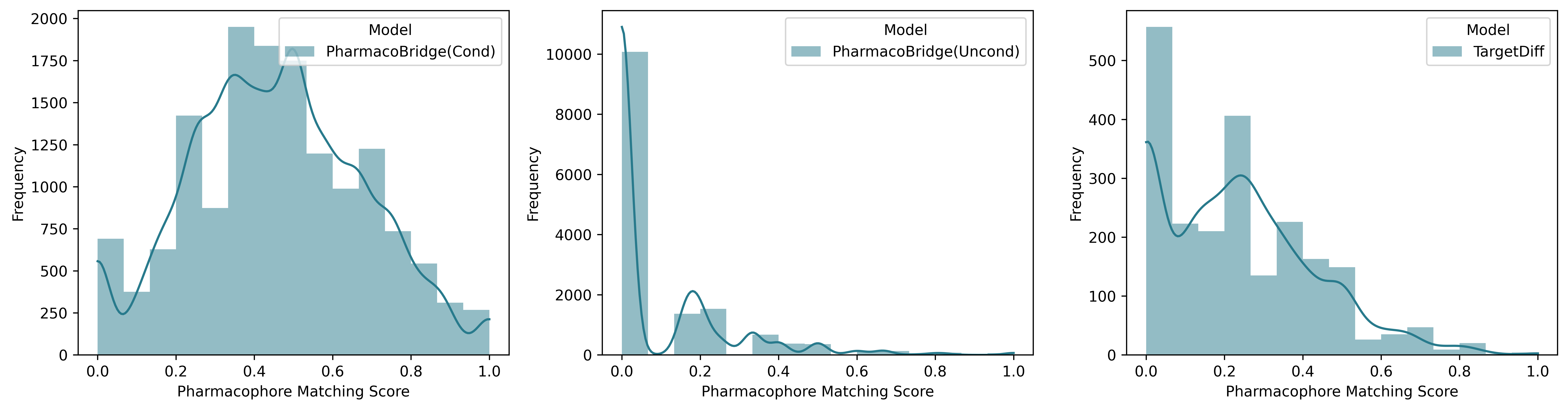}
    \caption{Pharmacophore matching sore distribution. Compared with Unconstrained generation and TargetDiff, pharmacophore-guided generation significantly enhanced the matching of pharmacophores extracted from original ligands and generated molecules. Notably, TargetDiff suffered from low validity, which resulting in less available molecules.}
    \label{fig:lig_pp}
\end{figure*}


Figure \ref{fig:sa} and \ref{fig:qed} present the distribution of SA and QED scores of generated molecules. Note that lower SA scores indicate better synthetic accessibility, and higher QED scores indicate better drug-likeness. Our model generated molecules with evenly distributed SA scores, whereas the SA scores of EDM and GruM generated molecules fell between 6.0 and 8.0, which is distant from the data distribution. In the QED assessment, GruM obtained relatively lower scores while our model and EDM generated molecules with better drug-likeness. Additional ablation study results on SA and QED evaluations are illustrated in \cref{sec:ablation}. 

In conclusion, PharmacoBridge achieved comparable validity, novelty and QED scores against state-of-the-art baselines and outperformed the baselines in terms of uniqueness and SA scores. In the ablation study, we found VE bridge-based PharmacoBridge was inclined to produce invalid chemical structures, whereas VP bridge-based PharmacoBridge exhibited better performance. This is aligned with the findings observed in previous works comparing VP and VE-based diffusion models \cite{zhou2023denoising, song2020score}. 
Moreover, incorporating aromatic features brought slight improvements in uniqueness and novelty. Therefore, we continue the subsequent experiments with VP-based PharmacoBridge with aromatic features. 

\subsection{Pharmacophore-guided Hit Molecule Design}

We further demonstrate PharmacoBridge in pharmacophore-guided drug design tasks, including ligand-based drug design and structure-based drug design (SBDD). 

\subsubsection{Ligand-based Drug Design}

Ligand-based drug design relies on the knowledge of active molecules whose target structures remain obscure. By analyzing the chemical properties and structures of these ligands, we can create pharmacophore hypotheses to identify or design new compounds with similar or improved interactions. In order to assess our model in the ligand-based drug design task, we sample new molecules using the pharmacophore features extracted from known active ligands. The pocket structure of each ligand's receptor protein is utilized as the constraint for sampling from Pocket2Mol and TargetDiff. 

We extracted all pharmacophores of around 15 thousand original ligands in the testing set, and removed the smaller groups that overlapped with the larger ones. New samples were generated by our model constrained by the extracted pharmacophores. We then computed the pharmacophore matching score reflecting the recovery rate of desired pharmacophore arrangements. Specifically, we calculated the distance between pharmacophores of generated molecules and original ligands. If the type of a pair of pharmacophores coincides and their distance is less than 1.5 \AA, we consider this as a match. The matching score of each ligand is calculated by $N_{match} / N_{total}$, where $N_{match}$ and $N_{total}$ represent the number of matches and the number of existing pharmacophores, respectively.

Figure \ref{fig:lig_pp} displays the histograms of pharmacophore matching scores. TargetDiff, which designs drugs constrained by pocket structures, were used as our baseline. Notably, only around two thousand molecules generated by TargetDiff were valid. Pocket2Mol was excluded from this assessment due to its limitation to autoregressive generation, which significantly prolongs the time required to generate approximately 15,000 samples without support for batch processing. In addition, we included molecules generated by the unconditional version of our model as one of the baselines for comparison. As shown in Figure \ref{fig:lig_pp}, PharmacoBridge significantly enhanced the pharmacophore recovery compared to unconstrained generation and TargetDiff. This suggests that the generated molecules have great potential to interact with the target protein, engaging not only within the desired pocket but also with the particular residue. 

\subsubsection{Structure-based Drug Design}

\begin{table}[ht]
    \centering
    \caption{Average pharmacophore matching scores of SBDD.}
    \label{tab:sbdd_pp}
    \begin{tabular}{cccc}
        \toprule
        PDB ID & PharmacoBridge & Pocket2Mol & TargetDiff \\ \midrule
        1EOC & 0.76 & 0.34 & 0.17 \\
        5LSA & 0.71 & 0.46 & 0.25 \\
        4H2Z & 0.91 & 0.49 & 0.38 \\
        1J06 & 1.00 & 0.08 & 0.11 \\
        3OCC & 0.74 & 0.22 & 0.24 \\
        5UEV & 0.84 & 0.39 & 0.17 \\
        5FE6 & 1.00 & 0.45 & 0.23 \\
        5LPJ & 0.83 & 0.41 & 0.17 \\
        5IJS & 0.97 & 0.31 & 0.14 \\
        4TTI & 0.98 & 0.37 & 0.20 \\ \bottomrule
    \end{tabular}
\end{table}

SBDD, in contrast to ligand-based drug design, utilizes the 3D structure of the target, typically a protein, to guide drug development. This method allows researchers to design molecules that fit precisely within the binding pocket of the target, optimizing interactions at them atomic level for enhanced efficacy and specificity. 

We selected 10 protein targets, manually identified the essential pharmacophore features for the ligands to bind with corresponding protein targets, and generated 100 hit candidates guided by each identified pharmacophore. Firstly, we evaluated the matching scores by comparing the pharmacophore features used as constraints and the ones extracted from the generated molecules. The average matching scores of PharmacoBridge and the baseline models are presented in Table \ref{tab:sbdd_pp}. Our method achieved the highest average matching score in all cases.


\begin{figure*}[ht]
    \centering
    \includegraphics[width=0.95\textwidth]{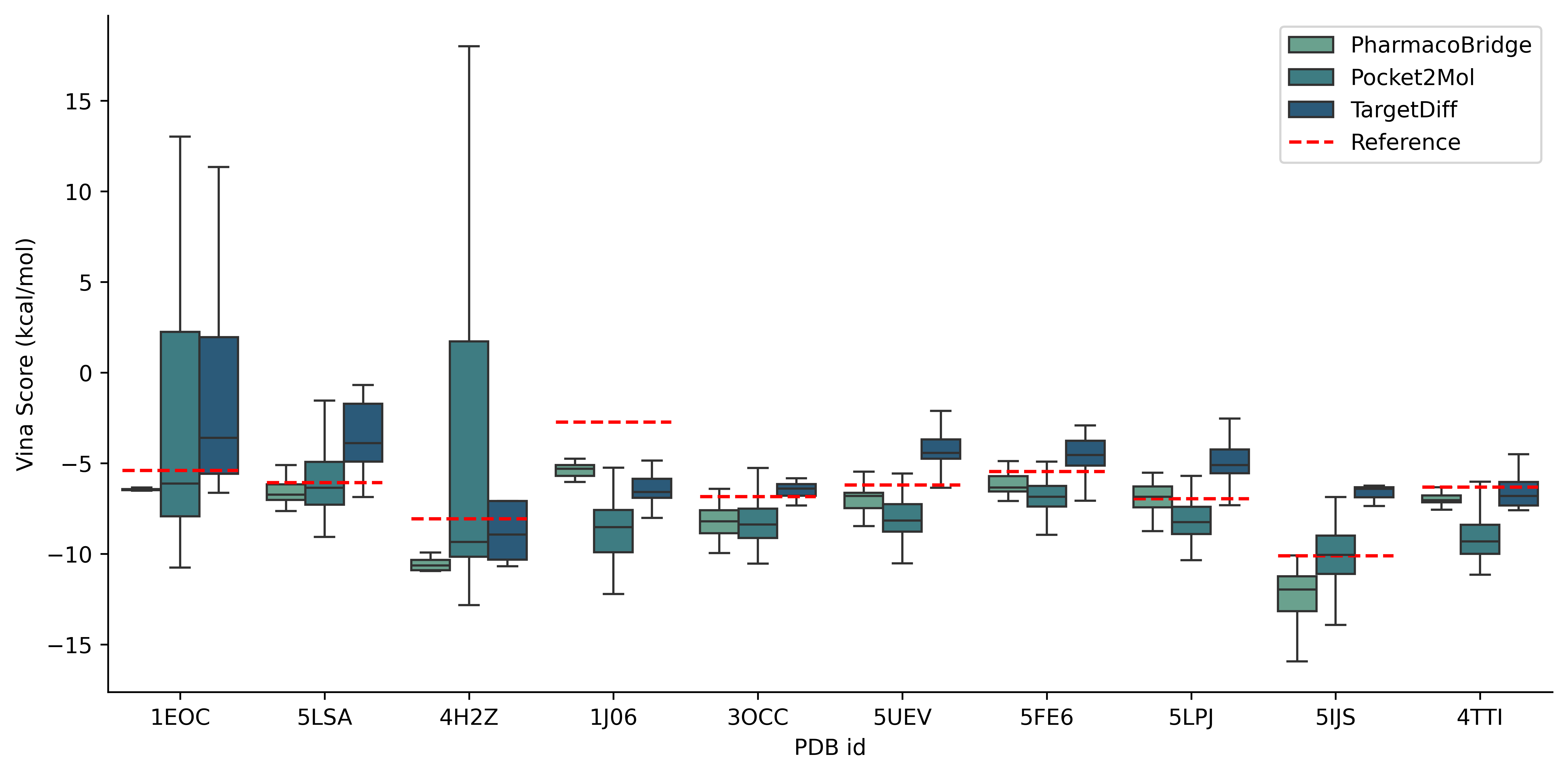}
    \caption{Distribution of Vina scores, with lower Vina score representing higher binding affinity. The reference scores given by the original ligands are indicated by red dashed lines in the figure. PharmacoBridge consistently generated molecules with higher binding affinities. Although Pocket2Mol sometimes outperformed, it produced molecules with a notably broader range of Vina scores in certain groups, e.g., 1EOC and 4H2Z. TargetDiff failed to generate molecules with binding affinities surpassing those of the original ligands in most cases.}
    \label{fig:ba_dist}
\end{figure*}

\begin{table*}[ht]
    \centering
    \caption{Percentage of generated molecules with higher binding affinity. }
    \label{tab:high_ba_ratio}
    \resizebox{\textwidth}{!}{
    \begin{tabular}{cccccccc}
        \toprule
        \multirow{2}{*}[-0.5\dimexpr \aboverulesep + \belowrulesep + \cmidrulewidth]{PDB ID} & \multicolumn{2}{c}{PharmacoBridge} & \multicolumn{2}{c}{Pocket2Mol} & \multicolumn{2}{c}{TargetDiff} & \multirow{2}{*}{Reference score} \\ \cmidrule(r){2-3} \cmidrule(r){4-5} \cmidrule(r){6-7}
         & High affinity ratio (\%) & Sample num & High affinity ratio (\%) & Sample num & High affinity ratio (\%) & Sample num & (kcal/mol) \\ \midrule
        1EOC & 99.00 & 100 & 55.34 & 100 & 30.00 & 10 & -5.39 \\
        5LSA & 76.00 & 100 & 52.78 & 100 & 3.70 & 28 & -6.06 \\
        4H2Z & 100.00 & 100 & 57.69 & 100 & 55.56 & 9 & -8.06 \\
        1J06 & 100.00 & 100 & 98.15 & 100 & 100.00 & 16 & -2.73 \\
        3OCC & 93.00 & 100 & 91.51 & 100 & 25.00 & 8 & -6.84 \\
        5UEV & 91.00 & 100 & 94.87 & 100 & 8.89 & 45 & -6.20 \\
        5FE6 & 80.00 & 100 & 90.99 & 100 & 18.75 & 32 & -5.46 \\
        5LPJ & 48.00 & 100 & 83.02 & 100 & 10.00 & 20 & -6.95 \\
        5IJS & 98.99 & 99 & 48.54 & 100 & 0.00 & 3 & -10.11 \\
        4TTI & 99.00 & 100 & 91.18 & 100 & 66.57 & 12 & -6.32 \\ \bottomrule
    \end{tabular}
    }
\end{table*}

Furthermore, we performed molecular docking analysis to demonstrate if the generated molecules are potential to interact with the protein targets. Gnina, a molecular docking software based on AutoDock Vina, was employed to test the binding affinity. Figure \ref{fig:ba_dist} shows the distribution of the Vina scores in box plots, with lower values representing higher binding affinities. The red dashed line in each group indicates the reference docking score of the original ligand provided by the CrossDocked dataset. It is evident that our method consistently generated molecules with higher binding affinities than the original ligand across each group. Pocket2Mol sometimes outperformed our method, but the generated molecules exhibited a significantly wide range of Vina scores, especially in cases of 1EOC and 4H2Z, rendering this approach very unstable. TargetDiff achieved the worst performance with most of the generated molecules gaining higher Vina scores than the reference ligand. 

In addition, the ratio of the generated molecules with higher binding affinity is calculated and shown in Table \ref{tab:high_ba_ratio}, together with the number of available samples and the reference Vina score. We observe that PharmacoBridge is able to produce hit candidates with over 90\% higher binding affinities in most cases. Pocket2Mol is competitive in certain groups such as 1J06, 3OCC, 5UEV, 5FE6 and 5IJS, but produced only around 50\% of higher affinity samples in the other groups. TargetDiff, however, is limited by a low validity rate. We present example molecules generated by PharmacoBridge and baselines using the pharmacophores identified in the binding complex structures of 1EOC, 3OCC and 5UEV in Figure \ref{fig:samples} in \cref{sec:samples}.

\section{Conclusion}
In this work, we presented PharmacoBridge - a diffusion bridge model designed for graph translation. Bridge processes enable us to construct a transitional probability flow between any two distributions and the idea of diffusion model makes such transition kernel tractable. We proposed a practical solution to leverage such models for graph generation constrained by the cluster arrangements of nodes, and demonstrated our model on the task of pharmacophore-guided \textit{de novo} hit molecule design. We parameterized our denoiser with $SE(3)$-equivariant GNN. Constrained by the pharmacophore models extracted from the active ligands, our model managed to generate potential hit candidates with higher binding affinity with the receptor of the original ligands.


\bibliography{example_paper}
\bibliographystyle{icml2025}

\newpage
\appendix
\onecolumn

\section{Proofs}

\subsection{Derivation of the Denoising Diffusion Bridge}
\label{sec:denoi_bridge_proof}

Applying Doob's $h$-transforms to the diffusion process with the law of \eqref{eq:diff_proc} results in a diffusion bridge solving Eq. \eqref{eq:diff_bridge}, i.e., $d\mathbf{G}_t = [\mathbf{f}(\mathbf{G}_t, t) + g(t)^2 \nabla_{\mathbf{G}_t}\log p(\mathbf{G}_T | \mathbf{G}_t)]dt + g(t)d\mathbf{w}_t$, which is pinned on both ends $\mathbf{G}_0 = \mathbf{g} \sim q_{data}(\mathbf{g})$ and $\mathbf{G}_T = \mathbf{\Gamma} \sim q_{data}(\mathbf{\Gamma})$. Pioneer works \cite{zhou2023denoising, peluchetti2023non} have proven that the time evolution of the conditional probability fixed at both ends $p(\mathbf{G}_t | \mathbf{G}_0 = \mathbf{g}, \mathbf{G}_T = \mathbf{\Gamma})$ follows the Fokker-Planck equation which amounts to the $h$-transformed diffusion bridge SDE \eqref{eq:diff_bridge}. Starting from here, we provide the proof of \cref{thm:denoisingbridge}. 
\begin{equation}
    \label{eq:diff_bridge_fp}
    \begin{aligned}
    &\frac{\partial}{\partial t} p(\mathbf{G}_t | \mathbf{G}_0 = \mathbf{g}, \mathbf{G}_T = \mathbf{\Gamma}) \\
    = &- \sum_{i=1}^d \frac{\partial}{\partial \mathbf{G}_t^i} \left[ \left(\mathbf{f}^i(\mathbf{G}_t, t) + g(t)^2 \sum_{j=1}^d \frac{\partial}{\partial \mathbf{G}_t^j} \log p(\mathbf{G}_T | \mathbf{G}_t) \right) p(\mathbf{G}_t | \mathbf{G}_0 = \mathbf{g}, \mathbf{G}_T = \mathbf{\Gamma}) \right] \\
    &+ \frac{1}{2} g^2(t) \sum_{i=1}^d \sum_{j=1}^d \frac{\partial ^2}{\partial \mathbf{G}_t^i \mathbf{G}_t^j} p(\mathbf{G}_t | \mathbf{G}_0 = \mathbf{g}, \mathbf{G}_T = \mathbf{\Gamma})
    \end{aligned}
\end{equation}

\denoisingbridge*
\begin{proof}
    As we design the sampling distribution to satisfy $q(\mathbf{G}_t | \mathbf{G}_0, \mathbf{G}_T) = p(\mathbf{G}_t | \mathbf{G}_0, \mathbf{G}_T)$ in Eq. \eqref{eq:q_dist}. The transition density associated with SDE \eqref{eq:denoi_bridge} equals to the expectation with respect to $\mathbf{G}_0$:
    \begin{equation}
    \label{eq:q_exp}
    \begin{aligned}
        q(\mathbf{G}_t | \mathbf{G}_T = \mathbf{\Gamma}) &= \int_{\mathbf{G}_0} q(\mathbf{G}_t | \mathbf{G}_0, \mathbf{G}_T = \mathbf{\Gamma}) q_{data}(\mathbf{G}_0 | \mathbf{G}_T = \mathbf{\Gamma}) d\mathbf{G}_0 \\
        &= \int_{\mathbf{G}_0} p(\mathbf{G}_t | \mathbf{G}_0, \mathbf{G}_T = \mathbf{\Gamma}) q_{data}(\mathbf{G}_0 | \mathbf{G}_T = \mathbf{\Gamma}) d\mathbf{G}_0 \\
        &= \mathbb{E}_{\mathbf{G}_0 \sim q_{data}(\mathbf{G}_0 | \mathbf{G}_T = \mathbf{\Gamma})} \left[ p(\mathbf{G}_t | \mathbf{G}_0, \mathbf{G}_T = \mathbf{\Gamma}) \right]
    \end{aligned}
    \end{equation}

    To investigate the time evolution of $q(\mathbf{G}_t | \mathbf{G}_T)$, we marginalize out $\mathbf{G}_0$ in the Fokker-Planck equation \eqref{eq:diff_bridge_fp}:
    \begin{equation}
    \label{eq:fp_exp}
    \begin{aligned}
        &\mathbb{E}_{\mathbf{G}_0 \sim q_{data}(\mathbf{G}_0 | \mathbf{G}_T = \mathbf{\Gamma})} \left[ \frac{\partial}{\partial t} p(\mathbf{G}_t | \mathbf{G}_0, \mathbf{G}_T = \mathbf{\Gamma}) \right] \\
        = &- \sum_{i=1}^d \frac{\partial}{\partial \mathbf{G}_t^i} \left[ \left(\mathbf{f}^i(\mathbf{G}_t, t) + g(t)^2 \sum_{j=1}^d \frac{\partial}{\partial \mathbf{G}_t^j}\log p(\mathbf{G}_T | \mathbf{G}_t) \right) \mathbb{E}_{\mathbf{G}_0 \sim q_{data}(\mathbf{G}_0 | \mathbf{G}_T = \mathbf{\Gamma})} \left[p(\mathbf{G}_t | \mathbf{G}_0, \mathbf{G}_T = \mathbf{\Gamma}) \right] \right] + \\
        &+ \frac{1}{2} g^2(t) \sum_{i=1}^d \sum_{j=1}^d \frac{\partial ^2}{\partial \mathbf{G}_t^i \mathbf{G}_t^j} \mathbb{E}_{\mathbf{G}_0 \sim q_{data}(\mathbf{G}_0 | \mathbf{G}_T = \mathbf{\Gamma})} [p(\mathbf{G}_t | \mathbf{G}_0, \mathbf{G}_T = \mathbf{\Gamma}]
    \end{aligned}
    \end{equation}
    
    Using the result in Eq. \eqref{eq:q_exp}, Eq. \eqref{eq:fp_exp} becomes
    \begin{equation}
        \label{eq:denoi_bridge_fp}
        \begin{aligned}
            &\frac{\partial}{\partial t} q(\mathbf{G}_t | \mathbf{G}_T = \mathbf{\Gamma})  \\
            = &- \sum_{i=1}^d \frac{\partial}{\partial \mathbf{G}_t^i} \left[ \left(\mathbf{f}^i(\mathbf{G}_t, t) + g(t)^2 \nabla_{\mathbf{G}_t}\log p(\mathbf{G}_T | \mathbf{G}_t) \right) q(\mathbf{G}_t | \mathbf{G}_T = \mathbf{\Gamma}) \right] + \frac{1}{2} g^2(t) \sum_{i=1}^d \sum_{j=1}^d \frac{\partial ^2}{\partial \mathbf{G}_t^i \mathbf{G}_t^j} q(\mathbf{G}_t | \mathbf{G}_T = \mathbf{\Gamma}) \\
            = &- \sum_{i=1}^d \frac{\partial}{\partial \mathbf{G}_t^i} \left[ \left( \mathbf{f}^i(\mathbf{G}_t, t) + g(t)^2 \nabla_{\mathbf{G}_t}\log p(\mathbf{G}_T | \mathbf{G}_t) - \frac{1}{2}g^2(t) \nabla_{\mathbf{G}_t}\log q(\mathbf{G}_t | \mathbf{G}_T = \mathbf{\Gamma}) \right) q(\mathbf{G}_t | \mathbf{G}_T = \mathbf{\Gamma}) \right]
        \end{aligned}
    \end{equation}
    which amounts to the Liouville equation (a special case of the Fokker-Planck equation when the diffusion coefficient is zero):
    \begin{equation}
    \label{eq:liouville}
    \begin{aligned}
        &\frac{\partial}{\partial t} q(\mathbf{G}_t | \mathbf{G}_T = \mathbf{\Gamma}) = - \sum_{i=1}^d \frac{\partial}{\partial \mathbf{G}_t^i} \left[ \Tilde{\mathbf{f}}^i(\mathbf{G}_t, t) q(\mathbf{G}_t | \mathbf{G}_T = \mathbf{\Gamma}) \right], \\
        \mathrm{where} \quad &\Tilde{\mathbf{f}}^i(\mathbf{G}_t, t) = \left( \mathbf{f}^i(\mathbf{G}_t, t) + g(t)^2 \nabla_{\mathbf{G}_t}\log p(\mathbf{G}_T | \mathbf{G}_t) - \frac{1}{2}g^2(t) \nabla_{\mathbf{G}_t}\log q(\mathbf{G}_t | \mathbf{G}_T = \mathbf{\Gamma}) \right)
    \end{aligned}
    \end{equation}
    This defines a probability flow ODE $d\mathbf{G}_t = \Tilde{\mathbf{f}}(\mathbf{G}_t, t) dt$ that is the same with Eq. \eqref{eq:denoi_bridge} specified in \cref{thm:denoisingbridge}.
\end{proof}

\subsection{$SE(3)$-equivariant Denoising Diffusion Bridge}
\label{sec:eddb_proof}

As proposed by \cite{luo2024bridging}, given a diffusion process solving $d\mathbf{G}_t = \mathbf{f}(\mathbf{G}_t, t)dt + g(t)d\mathbf{w}_t$, its associated transition density $p(\mathbf{G}_T | \mathbf{G}_t)$ is $SE(3)$-equivariant, i.e., $p(g \cdot \mathbf{G}_T | g \cdot \mathbf{G}_t) = p(\mathbf{G}_T | \mathbf{G}_t)$, $\forall g \in SE(3)$,
if the following conditions are satisfied: (1) the initial density $q(\mathbf{G}_0)$ is $SE(3)$-invariant; and (2) $\mathbf{f}(\cdot, t)$ is $SO(3)$-equivariant and $T(3)$-invariant; (3) the transition density of $\mathbf{w}_t$ is $SE(3)$-equivariant. 

In our case, the initial density $q(\mathbf{G}_0) = q_{data}(\mathbf{g}) = \delta(\mathbf{G}_0 - \mathbf{g})$ is a Dirac distribution that is invariant to rotations or translations. As we design the bridges using VP and VE schedules, the drift functions (as shown in Table \ref{tab:bridge_design}) are apparently $SO(3)$-equivariant. To meet $T(3)$-invariance, we calculate the center of mass of the system $(\mathbf{g}, \mathbf{\Gamma})$ and move the center to zero. Thus, we have a diffusion process with equivariant transition density $p(\mathbf{G}_T | \mathbf{G}_t)$. 

Based on this equivariant diffusion process, we can further devise the equivariant denoising diffusion bridge:
\eddb*

\begin{proof}
    We first investigate the equivariance of $p(\mathbf{G}_t | \mathbf{G}_0, \mathbf{G}_T)$, the marginal distribution of $\mathbf{G}_t$ with fixed endpoints. According to \cite{zhou2023denoising}, this distribution follows Bayes' rule 
    \begin{equation}
        \label{eq:bayes}
        p(\mathbf{G}_t | \mathbf{G}_0, \mathbf{G}_T) = \frac{p(\mathbf{G}_t | \mathbf{G}_0) p(\mathbf{G}_T | \mathbf{G}_t)}{p(\mathbf{G}_T | \mathbf{G}_0)}
    \end{equation}
    As we have obtained an $SE(3)$-equivariant transition density $p(\mathbf{G}_T | \mathbf{G}_t)$, when $t$ equals to zero, it is apparent that $p(\mathbf{G}_T | \mathbf{G}_0)$ is also $SE(3)$-equivariant. We therefore focus on investigating the equivariance of $p(\mathbf{G}_t | \mathbf{G}_0)$ now.

    Since each point cloud is defined by node position matrix and feature matrix, i.e., $\mathbf{G} = (\mathbf{x}, \mathbf{h})$, and the node features always remain invariant to rotations and translations, we only consider the effect of transformations with regard to $\mathbf{x}$ now. 
    The forward time evolution of $p(\mathbf{x}_t | \mathbf{x}_0)$ follows the Fokker-Planck equation:
    \begin{equation}
    \label{eq:diff_proc_fp}
    \begin{aligned}
        \frac{\partial}{\partial t}p(\mathbf{x}_t | \mathbf{x}_0) 
        = - \sum_{i=1}^d \frac{\partial}{\partial \mathbf{x}_t^i} \left[\mathbf{f}^i (\mathbf{x}_t, t) p(\mathbf{x}_t | \mathbf{x}_0) \right] + \frac{1}{2} g^2(t) \sum_{i=1}^d  \frac{\partial^2}{\partial \mathbf{x}_t^{i2}} p(\mathbf{x}_t | \mathbf{x}_0)
    \end{aligned}    
    \end{equation}
    
    For any transformation $g \in SE(3)$, let $\mathbf{y} = g \cdot \mathbf{x} = \mathbf{R}\mathbf{x} + \mathbf{t}$ denote the transformed atom positions, where $\mathbf{R} \in \mathbb{R}^{3 \times 3}$ is an orthogonal matrix defining the rotation and $\mathbf{t} \in \mathbb{R}^3$ defines the translation. The Fokker-Planck equation of the transformed data becomes
    \begin{equation}
    \label{eq:diff_proc_fp_trans}
    \begin{aligned}
        \frac{\partial}{\partial t}p(\mathbf{y}_t | \mathbf{y}_0) 
        = - \sum_{i=1}^d \frac{\partial}{\partial \mathbf{y}_t^i} \left[\mathbf{f}^i (\mathbf{y}_t, t) p(\mathbf{y}_t | \mathbf{y}_0) \right] + \frac{1}{2} g^2(t) \sum_{i=1}^d \frac{\partial^2}{\partial \mathbf{y}_t^{i2}} p(\mathbf{y}_t | \mathbf{y}_0)
    \end{aligned} 
    \end{equation}
    
    The partial derivative with respect to $\mathbf{y}$ follows
    \begin{equation}
    \label{eq:chain_rule_trans}
    \begin{aligned}
        \frac{\partial}{\partial \mathbf{y}_t^i} &= \sum_{j=1}^d  \frac{\partial \mathbf{x}_t^j}{\partial \mathbf{y}_t^i} \cdot \frac{\partial}{\partial \mathbf{x}_t^j} = \sum_{j=1}^d  \mathbf{R}_{ji}^{-1} \frac{\partial}{\partial \mathbf{x}_t^j} = \sum_{j=1}^d  \mathbf{R}_{ij} \frac{\partial}{\partial \mathbf{x}_t^j}, \\
        \mathrm{and} \quad \frac{\partial^2}{\partial \mathbf{y}_t^{i2}} &= \sum_{j=1}^d \sum_{k=1}^d \frac{\partial \mathbf{x}_t^j}{\partial \mathbf{y}_t^i} \cdot \frac{\partial}{\partial \mathbf{x}_t^j} \cdot \frac{\partial \mathbf{x}_t^k}{\partial \mathbf{y}_t^i} \cdot \frac{\partial}{\partial \mathbf{x}_t^k} = \sum_{j=1}^d \sum_{k=1}^d \mathbf{R}_{ij} \mathbf{R}_{ik} \frac{\partial}{\partial \mathbf{x}_t^j} \frac{\partial}{\partial \mathbf{x}_t^k}
    \end{aligned}
    \end{equation}
    Since $\mathbf{f}(\cdot, t)$ is designed to be $SO(3)$-equivariant and $T(3)$-invariant, we have
    \begin{equation}
    \label{eq:drift_trans}
        \mathbf{f}^i(\mathbf{y}_t, t) = \mathbf{f}^i(\mathbf{R}\mathbf{x}_t + \mathbf{t}, t) = \left(\mathbf{R}\mathbf{f}(\mathbf{x}_t, t) \right)_i = \sum_{k=1}^d \mathbf{R}_{ik} \mathbf{f}^k(\mathbf{x}_t, t)
    \end{equation}

    With Eq. \eqref{eq:chain_rule_trans} and \eqref{eq:drift_trans}, we can readily rewrite the transformed Fokker-Planck equation as
    \begin{equation}
    \label{eq:diff_proc_fp_trans2}
    \begin{aligned}
        \frac{\partial}{\partial t}p(\mathbf{y}_t | \mathbf{y}_0) = - \sum_{i=1}^d \sum_{j=1}^d \sum_{k=1}^d \mathbf{R}_{ij} \mathbf{R}_{ik} \frac{\partial}{\partial \mathbf{x}_t^j} \left[\mathbf{f}^k(\mathbf{x}_t, t) p(\mathbf{y}_t | \mathbf{y}_0) \right] + \frac{1}{2}g^2(t) \sum_{i=1}^d \sum_{j=1}^d \sum_{k=1}^d \mathbf{R}_{ij} \mathbf{R}_{ik} \frac{\partial}{\partial \mathbf{x}_t^j} \frac{\partial}{\partial \mathbf{x}_t^k} p(\mathbf{y}_t | \mathbf{y}_0)
    \end{aligned}
    \end{equation}
    Notably, since $\mathbf{R}$ is an orthogonal matrix, we have $\sum_{i=1}^d \mathbf{R}_{ij} \mathbf{R}_{ik} = \delta_{jk}$, which is zero when $j \neq k$. We further simplify \eqref{eq:diff_proc_fp_trans2} to
    \begin{equation}
    \label{eq:diff_proc_fp_trans3}
    \begin{aligned}
        \frac{\partial}{\partial t}p(\mathbf{y}_t | \mathbf{y}_0) &= - \sum_{j=1}^d \sum_{k=1}^d \delta_{jk} \frac{\partial}{\partial \mathbf{x}_t^j} \left[\mathbf{f}^k(\mathbf{x}_t, t) p(\mathbf{y}_t | \mathbf{y}_0) \right] + \frac{1}{2}g^2(t) \sum_{j=1}^d \sum_{k=1}^d \delta_{jk} \frac{\partial}{\partial \mathbf{x}_t^j} \frac{\partial}{\partial \mathbf{x}_t^k} p(\mathbf{y}_t | \mathbf{y}_0) \\
        &= - \sum_{j=1}^d \frac{\partial}{\partial \mathbf{x}_t^j} \left[\mathbf{f}^j(\mathbf{x}_t, t) p(\mathbf{y}_t | \mathbf{y}_0) \right] + \frac{1}{2}g^2(t) \sum_{j=1}^d \frac{\partial^2}{\partial \mathbf{x}_t^{j2}} p(\mathbf{y}_t | \mathbf{y}_0)
    \end{aligned}
    \end{equation}
    which characterizes the same time evolution as $p(\mathbf{x}_t | \mathbf{x}_0)$ in Eq. \eqref{eq:diff_proc_fp}. Moreover, the boundary condition $p(\mathbf{y} | \mathbf{y}_0) = \delta(\mathbf{y} - \mathbf{y}_0) = \delta(\mathbf{x} - \mathbf{x}_0) = p(\mathbf{x} | \mathbf{x}_0)$ is also $SE(3)$-equivariant. Therefore, the transition density $p(\mathbf{G}_t | \mathbf{G}_0)$ and the marginal distribution $p(\mathbf{G}_t | \mathbf{G}_0, \mathbf{G}_T)$ both satisfy $SE(3)$-equivariance, i.e., $p(g \cdot \mathbf{G}_t | g \cdot \mathbf{G}_0, g \cdot \mathbf{G}_T) = p(\mathbf{G}_t | \mathbf{G}_0, \mathbf{G}_T)$, $\forall g \in SE(3)$.

    Next, we look into the equivariance of the transition density associated with the denoising diffusion bridge $q(\mathbf{G}_t | \mathbf{G}_T)$. Focusing on the position matrix $\mathbf{x}$, we have $q(\mathbf{x}_t | \mathbf{x}_T) = \int_{\mathbf{x}_0} p(\mathbf{x}_t | \mathbf{x}_0, \mathbf{x}_T) q_{data}(\mathbf{x}_0 | \mathbf{x}_T) d\mathbf{x}_0$. Similarly, for any transformation $g \in SE(3)$, the transformed atom positions are denoted by $\mathbf{y} = g \cdot \mathbf{x} = \mathbf{R}\mathbf{x} + \mathbf{t}$. The transition density of $\mathbf{y}_t$ becomes
    \begin{equation}
    \label{eq:denoi_bridge_trans}
    \begin{aligned}
        q(\mathbf{y}_t | \mathbf{y}_T) &= \int p(\mathbf{y}_t | \mathbf{x}_0, \mathbf{y}_T) q_{data}(\mathbf{x}_0 | \mathbf{y}_T) d\mathbf{x}_0 \\
        &= \int p(g \cdot \mathbf{x}_t | g \cdot g^{-1} \cdot \mathbf{x}_0, g \cdot \mathbf{x}_T) q_{data}(g \cdot g^{-1} \cdot \mathbf{x}_0 | g \cdot \mathbf{x}_T) d\mathbf{x}_0
    \end{aligned}
    \end{equation}
    
    Since we have proven $p(g \cdot \mathbf{x}_t | g \cdot \mathbf{x}_0, g \cdot \mathbf{x}_T) = p(\mathbf{x}_t | \mathbf{x}_0, \mathbf{x}_T)$, and the distribution of the paired molecule and pharmacophore data also satisfies $q_{data}(g \cdot \mathbf{G}_0 | g \cdot \mathbf{G}_T) = q_{data}(\mathbf{G}_0 | \mathbf{G}_T)$, we have
    \begin{equation}
    \label{eq:denoi_bridge_trans2}
    \begin{aligned}
        q(\mathbf{y}_t | \mathbf{y}_T) = \int p(\mathbf{x}_t | g^{-1} \cdot \mathbf{x}_0, \mathbf{x}_T) q_{data}(g^{-1} \cdot \mathbf{x}_0 | \mathbf{x}_T) d\mathbf{x}_0
    \end{aligned}
    \end{equation}

    Let $\mathbf{x}_0' = g^{-1} \cdot \mathbf{x}_0$, then $\mathbf{x}_0 = g \cdot \mathbf{x}_0' = \mathbf{R}\mathbf{x}_0' + \mathbf{t}$. Therefore, 
    \begin{equation}
        \label{eq:derivative_conv}
        \frac{\partial \mathbf{x}_0}{\partial \mathbf{x}_0'} = \mathbf{R}, \quad and \quad \frac{d\mathbf{x}_0}{d\mathbf{x}_0'} = \text{det}(\mathbf{R}) = 1
    \end{equation}
    since $\mathbf{R}$ is an orthogonal matrix. 

    Notably, Eq. \eqref{eq:denoi_bridge_trans2} becomes
    \begin{equation}
    \label{eq:denoi_bridge_trans3}
    \begin{aligned}
        q(\mathbf{y}_t | \mathbf{y}_T) &= \int p(\mathbf{x}_t | \mathbf{x}_0', \mathbf{x}_T) q_{data}(\mathbf{x}_0' | \mathbf{x}_T) \text{det}(\mathbf{R}) d\mathbf{x}_0' \\
        &= \int p(\mathbf{x}_t | \mathbf{x}_0, \mathbf{x}_T) q_{data}(\mathbf{x}_0 | \mathbf{x}_T) d\mathbf{x}_0 \\
        &= q(\mathbf{x}_t | \mathbf{x}_T)
    \end{aligned}
    \end{equation}
    which completes the proof. 
\end{proof}

\section{Implementation details}

\subsection{Bridge Parameterization}
\label{sec:bridge_para}

PharmacoBridge supports flexible choices of the bridge design, such as the drift function $\mathbf{f(\cdot, t)}$ and the diffusion coefficient $g(\cdot)$. We here provide two typical designs, i.e., VP and VE bridges in Table \ref{tab:bridge_design}. For VP bridge, the variance schedule is given by $\beta_t = (\beta_{max} - \beta_{min})t + \beta_{min}$, and $\alpha_t$ is the integral of $\beta_t$: $\alpha_t = \int_{0}^{t} \beta_{\tau} d \tau $. For VE bridge, we have $\alpha_t = 1$ which is time-invariant.

For the choice of $\sigma_0$ and $\sigma_T$, we have performed grid search among $\{0.1n | n \in [1:10]\}$. Finally, we set $\sigma_{0, pos} = 0.1$, $\sigma_{T, pos} = 0.3$, $\sigma_{0, feat} = 0.7$ and $\sigma_{T, feat} = 1.0$. The covariance is calculated via $\sigma_{0T} = \sigma_0^2 / 2$. 

\begin{table}[h]
    \centering
    \caption{Design of VP and VE bridges}
    \label{tab:bridge_design}
    \begin{tabular}{cccccc}
    \toprule 
    & $\mathbf{f}\left(\mathbf{G}_t, t\right)$ & $g^2(t)$ & $p\left(\mathbf{G}_t \mid \mathbf{G}_0\right)$ & $\mathrm{SNR}_t$ & $\nabla_{\mathbf{G}_t}\log p(\mathbf{G}_T | \mathbf{G}_t)$ \\
    \midrule
    $\mathrm{VP}$ & $-\frac{1}{2} \beta_t \mathbf{G}_t$ & $\beta_t$ & $\mathcal{N}\left(\alpha_t \mathbf{G}_0, \sigma_t^2 \boldsymbol{I}\right)$ & $\alpha_t^2 / \sigma_t^2$ & $\frac{ (\alpha_t / \alpha_T) \mathbf{G}_T - \mathbf{G}_t }{ \sigma_t^2 ( \mathrm{SNR}_t / \mathrm{SNR}_T - 1) }$ \\
    $\mathrm{VE}$ & $\mathbf{0}$ & $\frac{d}{d t} \sigma_t^2$ & $\mathcal{N}\left(\mathbf{G}_0, \sigma_t^2 \boldsymbol{I}\right)$ & $1 / \sigma_t^2$ & $\frac{\mathbf{G}_T - \mathbf{G}_t}{\sigma_T^2 - \sigma_t^2}$ \\
    \bottomrule
    \end{tabular}
\end{table}

\subsection{Scaling Functions of Denoiser $D_{\theta}$} 
\label{sec:scaling}

Following DDBM \cite{zhou2023denoising}, our bridge scalings are designed as 
\begin{equation}
    \label{eq:scaling}
    \begin{aligned}
        &c_{in}(t) = \frac{1}{\sqrt{a_t^2 \sigma_T^2 + b_t^2 \sigma_t^2 + 2a_t b_t \sigma_{0T} + c_t}}, \\
        &c_{out}(t) = \sqrt{a_t^2(\sigma_T^2 \sigma_0^2 - \sigma_{0T}^2) + \sigma_0^2 c_t} * c_{in}(t), \\
        &c_{skip}(t) = (b_t \sigma_0^2 + a_t \sigma_{0T}) * c_{in}^2 (t), \\
        &c_{noise}(t) = \frac{1}{4} \log(\sigma_t), \\
        &\lambda(t) = \frac{1}{c_{out}(t)^2}, \\
        &\mathrm{where} \quad a_t = \frac{\alpha_t \text{SNR}_T}{\alpha_T \text{SNR}_t}, \quad b_t = \alpha_t(1 - \frac{\text{SNR}_T}{\text{SNR}_t}), \quad \mathrm{and} \quad c_t = \sigma_t^2 (1 - \frac{\text{SNR}_T}{\text{SNR}_t})
    \end{aligned}
\end{equation}
where $\sigma_0^2$ and $\sigma_T^2$ refer to the variance of the molecular graph $\mathbf{g}$ and the pharmacophore graph $\mathbf{\Gamma}$, respectively. 
The data variance at $t$-th step is configured as $\sigma_t = t$. The choice of $\sigma_0$ and $\sigma_t$ is discussed above in \cref{sec:bridge_para}. 

\subsection{Time Discretization}
\label{sec:time_disc}

Following EDM \cite{karras2022elucidating}, we discretize the time steps according to 
\begin{equation}
    \label{eq:time_disc}
        t_i = ( \sigma_{max}^{\frac{1}{\rho}} + \frac{N-i}{N-1} ( \sigma_{min}^{\frac{1}{\rho}} - \sigma_{max}^{\frac{1}{\rho}} ) )^{\rho}, \quad \mathrm{when} \quad i > 0 
\end{equation}
where $\rho$ is the parameter to control that shorter steps are taken near $\sigma_{min}$, and is set to $\rho = 7$ by default. We set $t_0 = 0$ to ensure the output is the reconstructed data sample. Empirically, for VE bridge, we set $\sigma_{min} = 0.02$ and $\sigma_{max} = 80$. And for VP bridge, we set $\sigma_{min} = 0.0001$ and $\sigma_{max} = 1$. 

\subsection{$SE(3)$-equivariant Network}
\label{sec:se3_model}

As illustrated in Eq. \eqref{eq:score_model} and \eqref{eq:edm}, the target of our score matching model is finalized to predict the mixture of signal and noise with $F_{\theta}$, which is, in our case, to predict the molecular point cloud with certain noise. A molecular point cloud is defined as $\mathbf{g} = (\mathbf{x}_{mol}, \mathbf{h}_{mol})$, where $\mathbf{x}_{mol} \in \mathbb{R}^{N \times 3}$ represents the coordinates of $N$ atoms in the 3D Euclidean space and $\mathbf{h}_{mol} \in \mathbb{R}^{N \times M}$ represents the node feature matrix that is composed of the one-hot encodings of $M$ atom types. A pharmacophore point cloud is defined as $\mathbf{\Gamma} = (\mathbf{x}_{pp}, \mathbf{h}_{pp})$. Similarly, $\mathbf{x}_{pp} \in \mathbb{R}^{N \times 3}$ records the center coordinates of the pharmacophore that each atom belongs to, and $\mathbf{h}_{pp} \in \mathbb{R}^{N \times K}$ is the feature matrix of the one-hot encodings of $K$ pharmacophore types. For data consistency, we perform zero padding to the feature matrix with less number of features, i.e., $\mathbf{h}_{mol}, \mathbf{h}_{pp} \in \mathbb{R}^{N \times F}$ with $F = \text{max}(M, K)$. Thus, the behavior of $F_{\theta}$ at time step $t$ is formulated as 
\begin{equation}
    \label{eq:f_theta}
    \begin{aligned}
        (\Tilde{\mathbf{x}}, \Tilde{\mathbf{h}}) &= F_{\theta} (c_{in}(t) \mathbf{G}_t, c_{noise}(t)) \\
        &= F_{\theta} (c_{in}(t) (\mathbf{x}_{t}, \mathbf{h}_{t}), c_{noise}(t))
    \end{aligned}
\end{equation}

\paragraph{Equivariance}
To manipulate 3D objects, it is crucial that the model $F_{\theta}$ satisfies $SE(3)$-equivariance. This ensures, for the coordinates $\mathbf{x}$, transformations such as rotations and translations of the input graph should result in equivariant output transformations, whereas the node features $\mathbf{h}$ should be invariant to any transformations. Assuming any transformation given by $\mathbf{Rx} + \mathbf{t}$ is undertaken to the input coordinates, an equivariant model should have
\begin{equation}
    \label{equivar}
    (\mathbf{R}\Tilde{\mathbf{x}} + \mathbf{t}, \Tilde{\mathbf{h}}) = F_{\theta} \left( c_{in}(t) (\mathbf{Rx}+ \mathbf{t}, \mathbf{h}), c_{noise}(t) \right)
\end{equation}
where $\mathbf{R} \in \mathbb{R}^{3 \times 3}$ is an orthogonal matrix defining the rotation and $\mathbf{t} \in \mathbb{R}^3$ defines the translation. 

To preserve $SE(3)$-equivariance, we build our model with EGNN \cite{satorras2021n}. To condition the updating of the molecular point cloud with the desired pharmacophore models, we concatenate the coordinates and features of the molecular point cloud at time $t$ and the initial pharmacophore point cloud to form a combined graph $\mathbf{G}_t = (\mathbf{x}_t, \mathbf{h}_t)$, where $\mathbf{x}_t \in \mathbb{R}^{2N \times 3}$ and $\mathbf{h}_t \in \mathbb{R}^{2N \times F}$. A mask $\mathcal{M}^{mol} \in \mathbb{R}^{2N}$ is applied to only update the nodes belonging to the molecular point cloud. 

\begin{figure}[ht]
    \centering
    \includegraphics[width=0.5\linewidth]{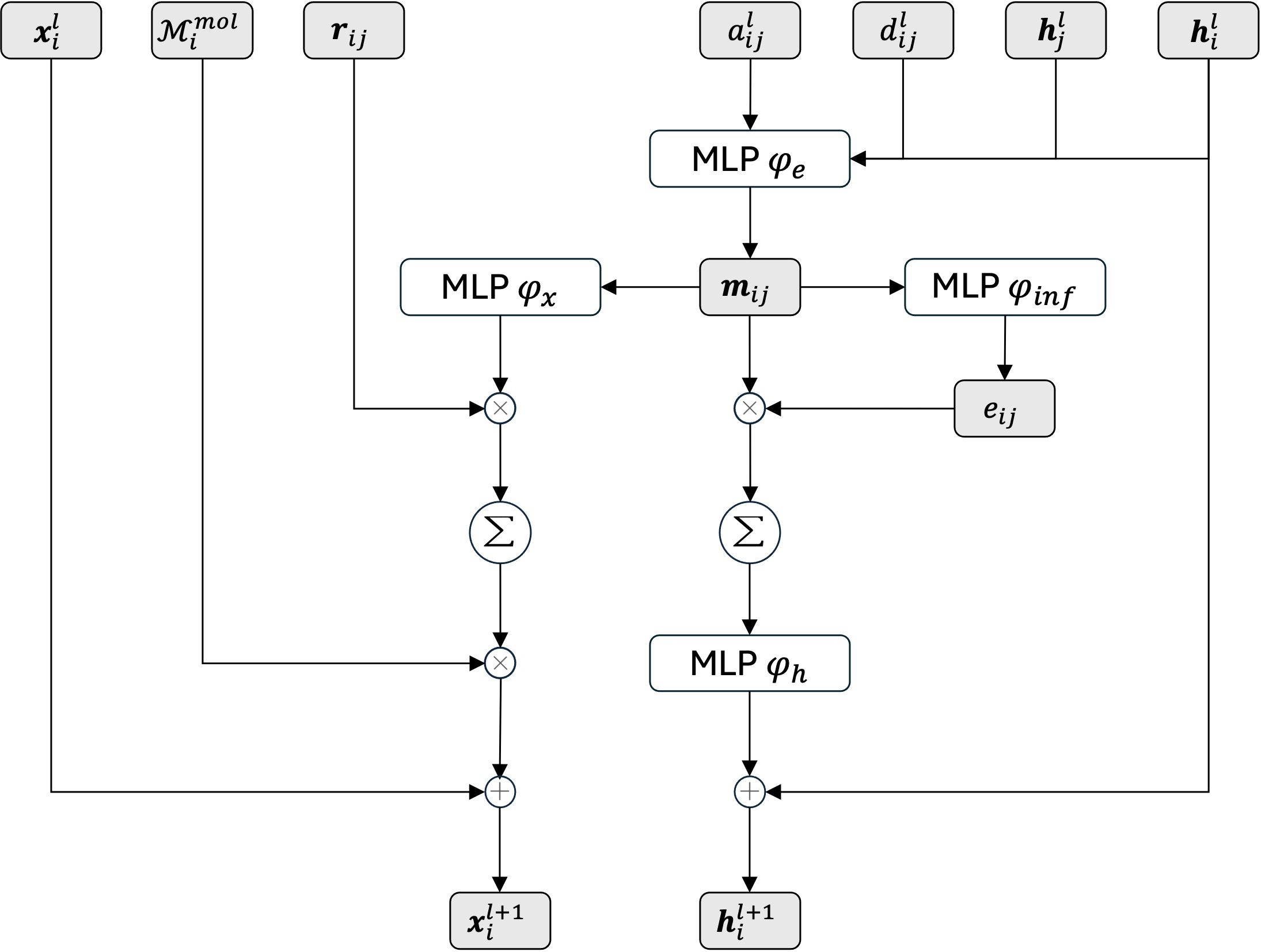}
    \caption{Flowchart of the EGCL module. $\otimes$ and $\oplus$ denote element-wise multiplication and addition, respectively. $\sum$ represents summation over all neighbors.}
    \label{fig:egcl}
\end{figure}

The Equivariant Graph Convolutional Layer (EGCL) that we used for feature updating is illustrated in Figure \ref{fig:egcl}. At $l$-th layer, the updating procedure is given by
\begin{equation}
    \label{eq:egnn}
    \begin{aligned}
        &\mathbf{m}_{ij} = \phi_e(\mathbf{h}_i^l, \mathbf{h}_j^l, d_{ij}^2, a_{ij}), \\
        &\mathbf{h}_i^{l+1} = \mathbf{h}_i^l + \phi_h(\mathbf{h}_i^l, \sum_{j \in \mathcal{N}(i)} e_{ij} \mathbf{m}_{ij} ), \\
        &\mathbf{x}_i^{l+1} = \mathbf{x}_i^l + \sum_{j \in \mathcal{N}(i)} \mathbf{r}_{ij} \phi_x(\mathbf{m}_{ij}) \cdot \mathcal{M}^{mol}_i
    \end{aligned}
\end{equation}
where $i$ and $j$ refer to the node index. $d_{ij} = \| \mathbf{x}_i - \mathbf{x}_j \|$ is the euclidean distance between node $i$ and $j$, and $\mathbf{r}_{ij} = \mathbf{x}_i^l - \mathbf{x}_j^l$ is the vector difference between node $i$ and $j$. $a_{ij}$ is the edge attribute indicating the interaction is between atoms, pharmacophore nodes, or an atom and a pharmacophore node. $e_{ij}$ is the soft estimation of the existence of edge between node $i$ and $j$, which is computed by $e_{ij} = \phi_{inf}(\mathbf{m}_{ij})$. Initially, the time step $t_i$ is incorporated with $\mathbf{h}_i^0$ via $\mathbf{h}_i^0 = \phi_{time} (\mathbf{h}_i^0, t_i)$.
$\phi_{time}$, $\phi_e$, $\phi_h$, $\phi_x$ and $\phi_{inf}$ are all trainable MLPs. 
As demonstrated by TargetDiff \cite{guan20233d}, EGCL governs the rotational equivariance. To fulfill translational equivariance as well, we must shift the center of the pharmacophore point cloud to zero. 

\subsection{Data Preprocessing}
\label{sec:data_preproc}

We extracted the pharmacophore and atom features from the ligands to form the node feature matrix. Typical pharmacophore features include Hydrophobe, Aromatic ring, Cation, Anion, Hydrogen Bond donor, Hydrogen Bond acceptor and Halogen. We considered atoms that is not in any pharmacophore belonging to an additional class, i.e., Linker. The pharmacophore feature matrix $\mathbf{h}_{pp}$ is thus composed of the one-hot encoding of each atom's pharmacophore cluster membership. In case of overlapping, we considered the atom belonging to the larger pharmacophore cluster with more atoms. 
For the atom features $\mathbf{h}_{mol}$, we considered two modes, i.e., basic mode and aromatic mode. In the basic mode, $\mathbf{h}_{mol}$ consists of the one-hot encoding of the atom type, including C, N, O, F, P, S and Cl. In the aromatic mode, $\mathbf{h}_{mol}$ consists of the one-hot encoding of the combination of the atom type and the aromatic property, which denotes if an atom appears in an aromatic ring. In this case, atom features contain C, C.ar, N, N.ar, O, O.ar, F, P, P.ar, S, S.ar and Cl. 

The position matrix $\mathbf{x}_{mol}$ is constructed using the coordinates of the atoms and $\mathbf{x}_{pp}$ is initialized as the center of pharmacophores adding a Gaussian noise with a small deviation to avoid having exactly the same destination for different nodes in the same pharmacophore. For atoms not belonging to any pharmacophore, we apply HDBSCAN to cluster the nodes that are spatially close and set $\mathbf{x}_{pp}$ as the center of the cluster. 

\section{Additional Results}

\subsection{Ablation Studies on SA and QED Evaluations}
\label{sec:ablation}

\begin{figure}[ht]
    \centering
    \includegraphics[width=0.9\textwidth]{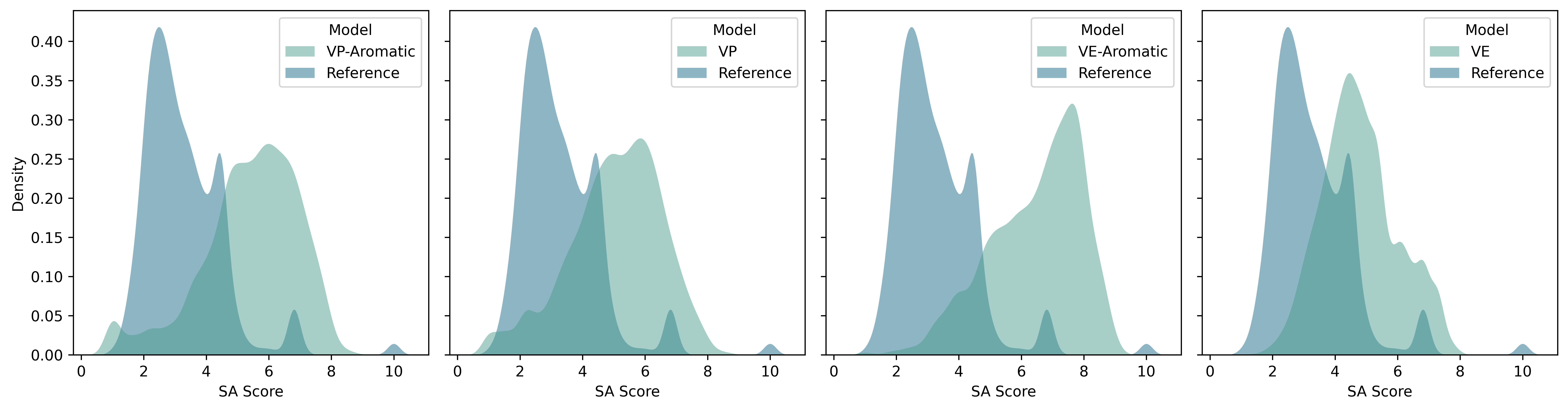}
    \caption{Ablation studies on SA analysis. PharmacoBridge with the VP design achieved similar SA scores with or without aromatic features. PharmacoBridge with the VE design without aromatic features achieved the most similar SA score distribution compared to the dataset distribution, but incorporating aromatic features significantly lower the synthetic accessibility. }
    \label{fig:sa_ablation}
\end{figure}

\begin{figure}[ht]
    \centering
    \includegraphics[width=0.9\textwidth]{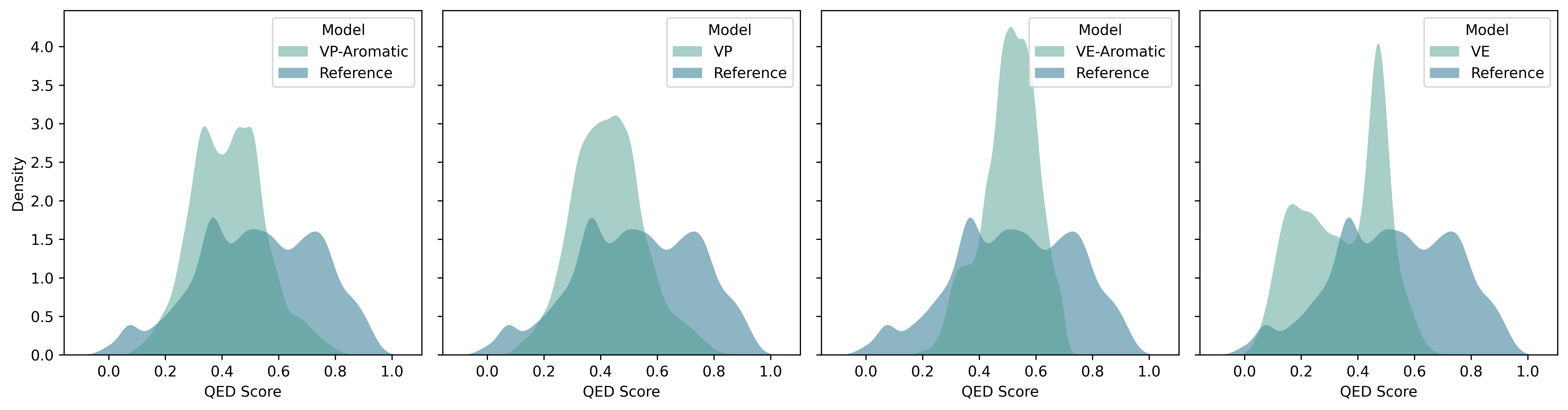}
    \caption{Ablation studies on QED analysis. Similar to SA analysis, PharmacoBridge with the VP design achieved comparative QED scores with or without aromatic features. PharmacoBridge with the VE design with aromatic features achieved the highest range of QED scores, but the QED scores dropped significantly without aromatic features. }
    \label{fig:qed_ablation}
\end{figure}

In this section, we investigate the effects of equipping PharmacoBridge with VP or VE schedules and the incorporation of aromatic features in SA and QED evaluations. 
Figure \ref{fig:sa_ablation} and Figure \ref{fig:qed_ablation} present the SA and QED score distributions of the ablation studies of our method. VP-based PharmacoBridge always achieved stable SA and QED score distributions with or without aromatic features. VE-based PharmacoBridge achieved lower SA scores (i.e., better synthetic accessibility) when aromatic features were absent, but obtained higher QED scores when aromatic features were present. Considering other ablation study results shown in Table \ref{tab:uncond}, we chose VP-based PharmacoBridge with aromatic features to continue subsequent conditional generation experiments. 

\subsection{Generated Samples}
\label{sec:samples}

\begin{figure}[ht]
    \centering
    \includegraphics[width=\textwidth]{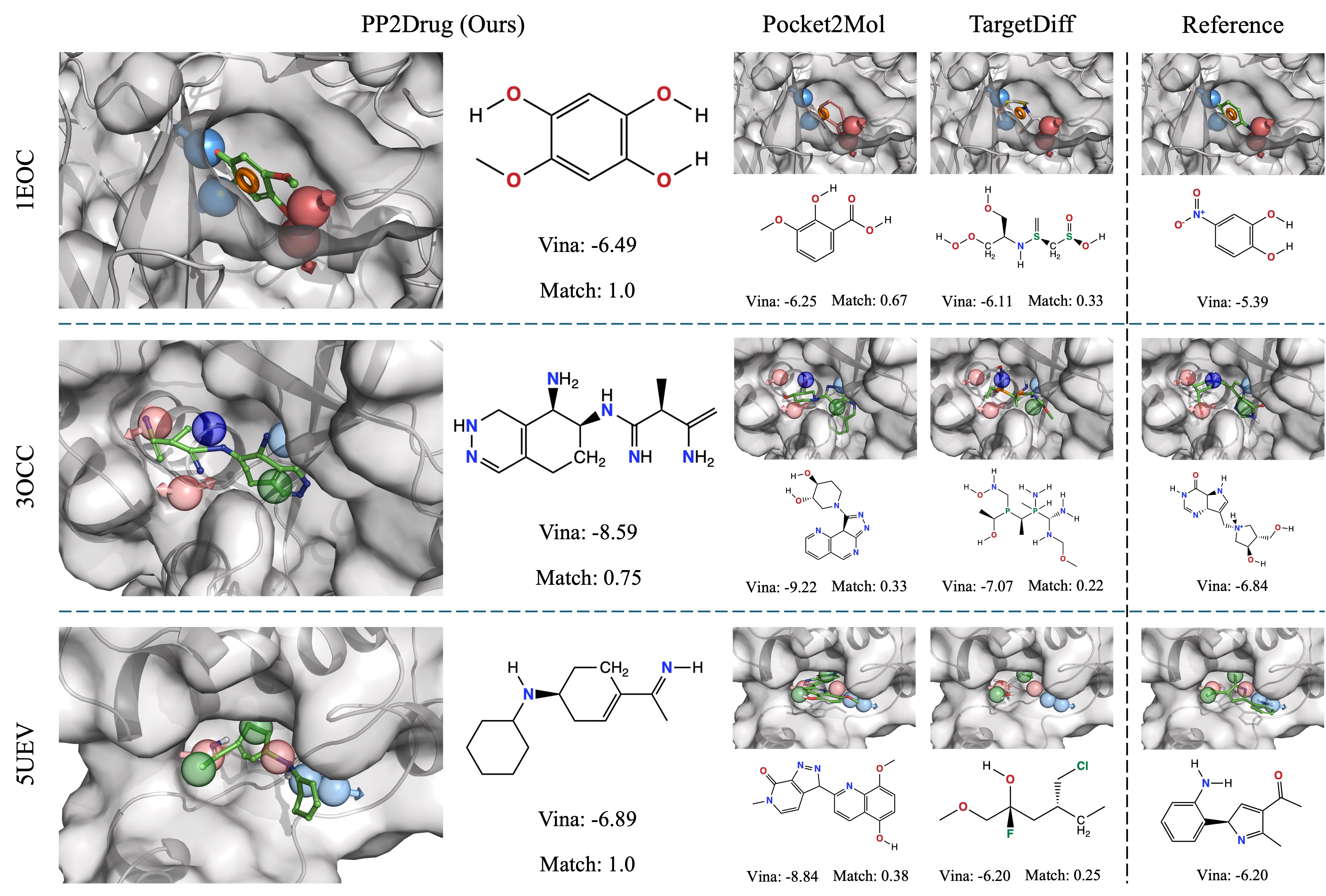}
    \caption{Molecules generated with the pharmacophore models of ligands from PDB structure 5JOY, 4CPI, 1DOD. Docked binding complexes and chemical structures, together with the Vina and CNN scores of both generated and original molecules are shown.}
    \label{fig:samples}
\end{figure}

Figure \ref{fig:samples} shows the reference molecules, and the molecules generated by PharmacoBridge and other baselines using the pharmacophores identified in the binding complex structures of 1EOC, 3OCC and 5UEV, including both 3D docked binding complexes and 2D chemical structures. 


\end{document}